\documentclass[10pt]{article}

\title{\bfseries Bounds on relative modular Hamiltonians\texorpdfstring{\\}{ }in general QFT}
\date{}
\newcommand\documentkeywords{relative entropy, modular Hamiltonian}

\usepackage{paper}  
\graphicspath{{figures/}} 

\usepackage[
backend=biber,
style=alphabetic,
sorting=nyt,
hyperref=true,
backref=false,
backrefstyle=none
]{biblatex}
\usepackage{authblk}
\author[1a]{Adriano Chialastri}
\author[2b]{Christoph Minz}
\author[2c]{Ko Sanders}
\affil[1]{SISSA, Mathematics Area, Via Bonomea 265, 34136 Trieste, Italy}
\affil[a]{\href{mailto:achialas@sissa.it}{achialas@sissa.it}}
\affil[2]{Leibniz Universität Hannover, Institut für Analysis, Welfengarten 1,
30167 Hannover, Germany}
\affil[b]{\href{mailto:christoph.minz@math.uni-hannover.de}{christoph.minz@math.uni-hannover.de}}
\affil[c]{\href{mailto:ko.sanders@math.uni-hannover.de}{ko.sanders@math.uni-hannover.de}}

\begin{document}

\maketitle

\begin{abstract}
The relative entropy between two states is a key concept in quantum information theory and quantum field theory. In the setting of quantum field theory, its computation requires the handling of relative modular Hamiltonians, which have nice general properties, but are typically very difficult to compute explicitly. In this paper, we exploit locality properties of general algebraic QFTs to devise a scheme that allows us to estimate relative modular Hamiltonians $K_{\tilde{\omega},\omega}$ between two states, $\omega$ and $\tilde{\omega}$, and hence also their relative entropy $H(\omega,\tilde{\omega})$, in terms of the modular Hamiltonian of a reference state $\hat{\omega}$, which might be better understood. For suitable pairs of states $(\omega,\tilde{\omega})$ we can estimate the relative modular Hamiltonian for the algebra of a region $V_2$ from above, resp.~from below, in terms of the modular Hamiltonian of $\hat{\omega}$ on a larger region $V_3\supset V_2$, resp.~a smaller region $V_1\subset V_2$.

Pairs of states and choices of regions which are susceptible to our scheme are related to the presence of superluminal signalling in the sense of Sorkin's paradox. In particular, if we choose $\omega=\hat{\omega}$ as our reference state, then there exists a unitary operation that maps $\omega$ to $\tilde{\omega}$ on $V_3$ and that does not allow superluminal signalling from the spacelike complement $V_3'$ to $V_2$, if our upper bound applies. Similarly, if our lower bound applies, then there is a unitary that does not allow superluminal signalling from $V_1$ to $V_2'$.

To investigate the strength of our estimates we consider coherent states for CCR systems, focussing particularly on free scalar fields. Our estimates apply even if the relative modular Hamiltonian cannot be computed exactly. For sufficiently regular excitations we can make the difference between the upper and lower bounds arbitrarily small and recover an exact result by squeezing. Our method thus yields an independent proof for the relative entropy formula in cases where the relative modular Hamiltonian cannot be computed exactly. For massless fields we establish the analogous result also for double cone regions. These results indicate that our estimates do not lose too much information. 
\end{abstract}

\section{Introduction}\label{sec:introduction}

Relative entropy is a central quantity in quantum information theory and in quantum field theory (QFT),  which serves as a notion of distance between states. Intuitively speaking, we interpret the relative entropy \(H(\omega,\tilde{\omega})\) as the amount of information that we expect to gain when we learn that our system is in the state \(\omega\), where we previously thought it was in the state \(\tilde{\omega}\)~\cite{BaezFritz2014}. Working in the framework of algebraic QFT, relative entropy can be defined very generally through Tomita-Takesaki modular theory, where it is obtained as a suitable expectation value of the relative modular Hamiltonian~\cite{Araki1975,Araki1977}. While their generality and nice properties make these objects compelling from a mathematical perspective, their drawback is that they are mostly very hard to compute explicitly. This naturally leads to the question whether relative modular Hamiltonians and relative entropies can be estimated in terms of other quantities. E.g., Bekenstein has argued on physical grounds that entropy can be estimated from above in terms of energies and energy densities using the Bekenstein bound
\cite{Bekenstein1981}. More recently, relative entropies have been estimated in terms of energy densities of a spacetime region of interest \cite{longo2025bekenstein,hollands2025bekenstein}. In this paper, we will devise a scheme that allows us to estimate relative modular Hamiltonians, and hence also relative entropies, from above, resp.~from below, in terms of modular Hamiltonians on larger, resp.~smaller, regions. Our scheme combines locality properties of general QFTs with general properties of relative modular operators and suitable assumptions on the states involved.

Modular theory can be defined for pairs of states on a von Neumann algebra \cite{Araki1975,Araki1977,TakesakiII}, leading to relative modular Hamiltonians. Modular Hamiltonians arise as a special case, where both states are equal. This special case is more commonly presented in the literature, because it is easier and it can be generalised to the setting of standard subspaces of complex Hilbert spaces~\cite{longo2008real}. In QFT, whenever explicit expressions are known, they usually concern the modular Hamiltonian of a single state.

Unfortunately, explicit results for modular Hamiltonians in QFT are only known in a handful of cases. All these cases have in common that the modular flow generated by the modular Hamiltonian has a geometric action. Most famous is the Bisognano-Wichmann Theorem, which concerns the vacuum state restricted to a Rindler wedge in Minkowski space~\cite{BisognanoWichmann1976,BisognanoWichmann:1975}. This result can be generalised to curved spacetimes \cite{SEWELL1982201,sanders2015construction}. Further results concern massless free scalar fields in the Minkowski vacuum for double cones, cf.~\cite{HislopLongo1982} for the modular flow and \cite{LongoMorsella2023} for the modular Hamiltonian, and for the future lightcone~\cite{buchholz1978},~\cite[Thm.~V.4.2.2]{haag2012local}. For massless fermions, the modular Hamiltonian of a double cone was determined in \cite{la2025fermionic}.

It is of interest to see if it is possible to express or estimate relative modular Hamiltonians for some states $\omega$ and $\tilde{\omega}$ in terms of the modular Hamiltonian of some reference state $\hat{\omega}$, which might be better understood. Such relations are bound to depend on the operators necessary to link $\omega$ and $\tilde{\omega}$ to $\hat{\omega}$. In Section~\ref{sec:GeneralResults}, we will exploit the locality properties of general algebraic QFTs to describe our scheme, which,
for suitable pairs of states \(\omega\) and \(\tilde{\omega}\), allows us to estimate the relative modular Hamiltonian \(K_{\tilde{\omega},\omega}\) of some region \(V_2\), and hence also the relative entropy \(H(\omega,\tilde{\omega})\), from above, resp.~from below, in terms of the modular Hamiltonian of the reference state \(\hat{\omega}\) on a larger region \(V_3\supset V_2\), resp.~a smaller region \(V_1\subset V_2\).

To examine which pairs of states and choices of regions are susceptible to our scheme, it is interesting to note that its applicability is linked to Sorkin's paradox. In~\cite{Sorkin:1993}, Sorkin points out that it is possible to construct examples of superluminal signalling, so-called \enquote{impossible measurements}. In QFT, the typical setting involves three observers, Alice, Bob and Charlie, such that Alice and Charlie are localised in spacelike related regions. An impossible measurement arises if Bob can perform a unitary operation that allows Charlie to determine whether Alice has performed some unitary operation on the shared quantum system or not, apparently violating causality. In Section~\ref{ssec:applicationinQFT} we will adopt the point of view of \cite{borsten2021impossible,jubb2022causal}, where conditions on the unitaries localised in Bob's part of the system were formulated that ensure they are non-signalling, i.e., they do not allow superluminal signalling. We refer to~\cite{papageorgiou2024eliminating} for a comprehensive review of the current state of research on Sorkin's paradox.

Regarding our estimates, we observe that if we choose \(\omega=\hat{\omega}\) as our reference state, then there is a unitary operation that maps \(\omega\) to \(\tilde{\omega}\) on \(V_3\) and that does not allow superluminal signalling from the spacelike complement \(V_3'\) to \(V_2\) in the sense of Sorkin's paradox, if our upper bound applies. Similarly, if our lower bound applies, then there exists a unitary that maps \(\omega\) to \(\tilde{\omega}\) on \(V_2\) and that does not allow superluminal signalling from \(V_1\) to \(V_2'\). However, it is an open question whether the existence of a non-signalling unitary (in either direction) is a sufficient condition to show the applicability of our estimate (in the corresponding direction). We comment on this issue using a construction involving the Cuntz algebra.

To investigate the strength of our estimates we consider applications to coherent states, first on general CCR Weyl algebras and then for free scalar fields in the Minkowski vacuum representation. In general, coherent states in a quasi-free representation are defined in terms of Weyl operators, which are always non-signalling unitaries. For a standard reference state, the coherent states are also standard, so that relative modular Hamiltonians are always well-defined. One may often express the relative modular Hamiltonian between two coherent states, and hence also their relative entropy, in terms of the modular Hamiltonian of the quasi-free reference state. However, for coherent states whose excitations have a non-zero difference at the boundary of the region the exact computation becomes non-trivial. This is why our estimates are of some interest even in this simple case.

It is known that the relative entropy between all pairs of coherent states can be obtained by approximation arguments using its lower semicontinuity \cite{Bostelmann2022,Longo:2019,Ciolli2020}. (This argument does not directly apply to the relative modular Hamiltonian.) For the special case of free scalar fields on a Rindler wedge or, in the massless case, on a double cone we show that an exact formula for the relative entropy can also be recovered from our estimates using a squeezing argument involving an appropriate choice of cutoff functions. This proof works also for excitations that do not vanish at the boundary and it differs from existing proofs, because it does not exploit the lower semicontinuity of the relative entropy. This result indicates that our general estimates may not be losing too much information.

We have organised our paper as follows. In Section~\ref{sec:GeneralResults} we review some generalities about relative modular theory and relative entropy in QFT, establish our upper and lower bounds and discuss the relation between the applicability of these estimates to Sorkin's paradox. Section~\ref{sec:coherentstates} applies these results to coherent states in general Weyl CCR algebras. Section~\ref{sec:coherentscalarQFT} will further specialise the setting to coherent excitations of the Minkowski vacuum of a free scalar field in order to present further results on relative modular Hamiltonians and to show that our upper and lower bounds for the relative entropy on a Rindler wedge or on a double cone (in the massless case) can be chosen arbitrarily close to the exact result. Section~\ref{sec:conclusions} contains our conclusions and a brief outlook.

\section{Relative modular theory and relative entropy in QFT}\label{sec:GeneralResults}

In this section we review some facts about relative modular theory and relative entropy in QFT.  We refer the reader to \cite{TakesakiI,TakesakiII,KadisonRingroseI,KadisonRingroseII} for background material on operator algebras and to \cite{HollandsKS2018} for more details on relative entropy in QFT and its relation to entanglement entropy.

\subsection{Generalities on relative modular Hamiltonians and relative entropy}\label{sec:ModularGeneralities}

Given a $C^{\star}$-algebra $\Acal$ with unit $\one$ and a state $\omega$ (i.e. a normalised positive linear functional $\omega:\Acal\to\Cbb$), its GNS-representation is characterised by the GNS-triple $(\pi_{\omega},\Hcal_{\omega},\Omega_{\omega})$ and we define the von Neumann algebra $\Rcal_{\omega} \coloneq \pi_{\omega}(\Acal)''$. We call $\Omega_{\omega}$ a standard vector for $\Rcal_{\omega}$ iff it is both separating and cyclic, i.e., $\Rcal_{\omega}\ni a\mapsto a\Omega_{\omega}\subset\Hcal_{\omega}$ is injective and has a dense range.

Now, let $\tilde{\omega}$ be another state on $\Acal$ with GNS-triple $(\pi_{\tilde{\omega}},\Hcal_{\tilde{\omega}},\Omega_{\tilde{\omega}})$.
The states $\omega$ and $\tilde{\omega}$ are called quasi-equivalent iff the von Neumann algebras $\Rcal_{\omega}$ and $\Rcal_{\tilde{\omega}}$ are isomorphic.
If in addition both $\Omega_{\omega}$ and $\Omega_{\tilde{\omega}}$ are standard vectors, then the Unitary Implementation Theorem \cite[Thm.~7.2.9]{KadisonRingroseII} states that we can identify $\Hcal \coloneq \Hcal_{\tilde{\omega}} = \Hcal_{\omega}$ and $\pi_{\tilde{\omega}}=\pi_{\omega}$, so the vectors $\Omega_{\omega}$ and $\Omega_{\tilde{\omega}}$ lie in the same Hilbert space. This is the situation that we will now consider and for simplicity we will denote the vectors by $\Omega$ and $\tilde{\Omega}$ and drop the subscripts $\omega$ and $\tilde{\omega}$.

Given two standard unit vectors $\Omega,\tilde{\Omega}\in\Hcal$ for the von Neumann algebra $\Rcal$ one defines the relative Tomita operator $S_{\tilde{\Omega},\Omega}$ by
\begin{align}
S_{\tilde{\Omega},\Omega}a\Omega&\coloneq a^*\tilde{\Omega}\label{eqn:defTomita}
\end{align}
for all $a\in\Rcal$. This operator is densely defined with
\begin{align}
S_{\tilde{\Omega},\Omega}^{-1}&=S_{\Omega,\tilde{\Omega}}
\eqend{.}\label{eqn:Sinverse}
\end{align}
We can analogously define the operator 
$S'_{\tilde{\Omega},\Omega}$ w.r.t.~$\Rcal'$, for which $\Omega$ and $\Omega'$ are also standard. An elementary computation then shows that 
$S_{\tilde{\Omega},\Omega}^*\supset S'_{\tilde{\Omega},\Omega}$. In particular, 
$S_{\tilde{\Omega},\Omega}^*$ is densely defined and hence $S_{\tilde{\Omega},\Omega}$ is closable. We will denote the closure by the same symbol and we note that \eqref{eqn:Sinverse} remains valid.

The polar decomposition
\begin{align}
S_{\tilde{\Omega},\Omega}&=J_{\tilde{\Omega},\Omega}\Delta_{\tilde{\Omega},\Omega}^{\frac12}\notag
\end{align}
uniquely defines the relative modular operator $\Delta_{\tilde{\Omega},\Omega}$, which is strictly positive, and the relative modular conjugation $J_{\tilde{\Omega},\Omega}$, which is anti-unitary, $J_{\tilde{\Omega},\Omega}^*=J_{\tilde{\Omega},\Omega}^{-1}$. The relative modular Hamiltonian is the self-adjoint operator
\begin{align}
K_{\tilde{\Omega},\Omega}&\coloneq -\log\left(\Delta_{\tilde{\Omega},\Omega}\right)\eqend{.}\label{def:relK}
\end{align}

From \eqref{eqn:Sinverse} we find $S_{\Omega,\tilde{\Omega}}=J_{\tilde{\Omega},\Omega}^*(J_{\tilde{\Omega},\Omega}\Delta_{\tilde{\Omega},\Omega}^{-\frac12}J_{\tilde{\Omega},\Omega}^*)$
so by the uniqueness of the polar decomposition we have
\begin{align}
J_{\Omega,\tilde{\Omega}}&=J_{\tilde{\Omega},\Omega}^*
\notag\\
\Delta_{\Omega,\tilde{\Omega}}^{\frac12}&=
J_{\tilde{\Omega},\Omega}
\Delta_{\tilde{\Omega},\Omega}^{-\frac12}
J_{\tilde{\Omega},\Omega}^*\notag\\
K_{\Omega,\tilde{\Omega}}&=-
J_{\tilde{\Omega},\Omega}
K_{\tilde{\Omega},\Omega}
J_{\tilde{\Omega},\Omega}^*\eqend{.}\notag
\end{align}
For the adjoint we then have $S_{\tilde{\Omega},\Omega}^*=J_{\Omega,\tilde{\Omega}}\Delta_{\Omega,\tilde{\Omega}}^{-\frac12}$.

The following lemma is sometimes useful, cf.
\cite[Lemma 5.1]{CasiniGrilloPontello:2019}.
\begin{lemma}\label{lem:UModular}
If a unitary $u$ acting on $\Hcal$ satisfies $u\Rcal u^*=\Rcal$, then $u\Omega$ and $u\tilde{\Omega}$ are also standard and we have
\begin{align}
J_{u\tilde{\Omega},u\Omega}&=uJ_{\tilde{\Omega},\Omega}u^*\notag\\
K_{u\tilde{\Omega},u\Omega}&=uK_{\tilde{\Omega},\Omega}u^*\eqend{.}\notag
\end{align}
\end{lemma}
By functional calculus it then follows that $\Delta_{u\tilde{\Omega},u\Omega}=u\Delta_{\tilde{\Omega},\Omega}u^*$ and 
$S_{u\tilde{\Omega},u\Omega}=uS_{\tilde{\Omega},\Omega}u^*$.
\begin{proof}
As $\Rcal u\Omega=u\Rcal\Omega$,  $u\Omega$ is cyclic for $\Rcal$. Using $u\Rcal'u^*=\Rcal'$ we see from the same argument that $u\Omega$ is also cyclic for $\Rcal'$ and hence separating and standard for $\Rcal$. The same applies to $u\tilde{\Omega}$. Now $\Rcal u\Omega=u\Rcal\Omega$ is a core for both $S_{u\tilde{\Omega},u\Omega}$ and $uS_{\tilde{\Omega},\Omega}u^*$ and for all $a\in\Rcal$ we have
\begin{align}
uS_{\tilde{\Omega},\Omega}u^*au\Omega&=
u(u^*au)^*\tilde{\Omega}
=a^*u\tilde{\Omega}
=S_{u\tilde{\Omega},u\Omega}au\Omega\eqend{,}\notag
\end{align}
so $S_{u\tilde{\Omega},u\Omega}=uS_{\tilde{\Omega},\Omega}u^*=(uJ_{\tilde{\Omega},\Omega}u^*)(u\Delta^{\frac12}_{\tilde{\Omega},\Omega}u^*)$, from which we infer all the desired relations using the uniqueness of the polar decomposition.
\end{proof}

When $\tilde{\Omega}=\Omega$ we drop the adjective \enquote{relative} and we use the single subscript $\Omega$, so the Tomita operator $S_{\Omega}$, the modular operator $\Delta_{\Omega}$, the modular conjugation $J_{\Omega}$ and the modular Hamiltonian $K_{\Omega}$ satisfy the relations $S_{\Omega}=J_{\Omega}\Delta_{\Omega}^{\frac12}$ and $K_{\Omega}=-\log(\Delta_{\Omega})$ with
\begin{align}
J_{\Omega}&=J_{\Omega}^*=J_{\Omega}^{-1}\notag\\
K_{\Omega}&=-J_{\Omega}K_{\Omega}J_{\Omega}\eqend{.}\notag
\end{align}

Following \cite{Araki1975,Araki1977}, the relative entropy between the states $\Omega$ and $\tilde{\Omega}$ for the algebra $\Rcal$ is defined by\footnote{Unfortunately the corresponding Equation (91) in \cite{HollandsKS2018} suffers from a sign error and swapped subscripts.}
\begin{align}
H(\Omega,\tilde{\Omega})&\coloneq 
\langle\Omega,K_{\tilde{\Omega},\Omega}\Omega\rangle
=-\langle\Omega,
\log \,\!(\Delta_{\tilde{\Omega},\Omega})\Omega\rangle\eqend{.}\label{eqn:defRE}
\end{align}
We interpret this as the amount of information that we expect to gain when we learn that our system is in the state $\Omega$, where we previously thought it was in the state $\tilde{\Omega}$ \cite{BaezFritz2014}. Note that $H(\Omega,\tilde{\Omega})\not=H(\tilde{\Omega},\Omega)$ in general.

The following standard example expresses the relative entropy in terms of density matrices.
\begin{example}
Let $\Hcal_1$ be a separable complex Hilbert space and $\rho,\tilde{\rho}$ two density matrices, i.e. positive operators with trace $1$. We let $\Hcal$ be the Hilbert space of Hilbert-Schmidt operators on $\Hcal_1$, endowed with the Hilbert-Schmidt inner product $\langle a,b\rangle \coloneq \tr_{\Hcal_1}a^*b$.
We will write $\Omega \coloneq \sqrt{\rho}$ and $\tilde{\Omega} \coloneq \sqrt{\tilde{\rho}}$ as vectors in $\Hcal$.
We may identify $\Hcal\simeq\Hcal_1\otimes\Hcal_1^*$ in a natural way, where $\Hcal_1^*$ is the dual space of $\Hcal_1$.
E.g., if $\{v_n\}$ is an orthonormal eigenbasis for $\rho$ with eigenvalues $p_n\ge0$, then $\Omega=\sum_n\sqrt{p_n}v_n\otimes v_n^*$.

We let $\Rcal=\Bcal(\Hcal_1)$ the algebra of all bounded operators on $\Hcal_1$. This algebra acts on the first factor of $\Hcal_1\otimes\Hcal_1^*$. The commutant $\Rcal'=\Bcal(\Hcal_1^*)$ acts on the second factor. If $\rho$ and $\tilde{\rho}$ are strictly positive, the vectors $\Omega$ and $\tilde{\Omega}$ are cyclic and separating for $\Rcal$ and we can define the relative Tomita operator. We will verify that the relative modular operator is given by
\begin{align}
\Delta_{\tilde{\Omega},\Omega}&=
\tilde{\rho}\otimes\rho^{-1}\notag
\end{align}
acting on $\Hcal_1\otimes\Hcal_1^*$ by
$\Delta_{\tilde{\Omega},\Omega}v\otimes w^*=(\tilde{\rho}v)\otimes(\rho^{-1}w)^*=\tilde{\rho}(v\otimes w^*)\rho^{-1}$. Indeed, for all $a,b\in\Rcal$ we have
\begin{align*}
    \bigl\langle
      a \Omega,
      \Delta_{\tilde{\Omega}, \Omega} b\Omega
    \bigr\rangle
  &= \bigl\langle
      S_{\tilde{\Omega},\Omega} b\Omega,
      S_{\tilde{\Omega},\Omega} a \Omega
    \bigr\rangle
\nexteq
  &= \bigl\langle b^* \tilde{\Omega}, a^*\tilde{\Omega} \bigr\rangle
\nexteq
  &= \tr_{\Hcal_1}\bigl( \tilde{\rho} ba^* \bigr)
\nexteq
  &= \tr_{\Hcal_1}\Bigl(
      \rho
      \bigl( \sqrt{\tilde{\rho}} a \rho^{-\frac12} \bigr)^*
      \bigl( \sqrt{\tilde{\rho}} b \rho^{-\frac12} \bigr)
    \Bigr)
\nexteq
  &= \Bigl\langle
      \sqrt{\tilde{\rho}} \otimes \rho^{-\frac12} a \Omega,\sqrt{\tilde{\rho}} \otimes \rho^{-\frac12} b \Omega
    \Bigr\rangle
\nexteq
  &= \bigl\langle a \Omega, \tilde{\rho} \otimes \rho^{-1} b \Omega \bigr\rangle
  \eqend{,}
\end{align*}
where the first and last equalities hold in the sense of quadratic forms. From this we conclude that
\begin{align*}
    K_{\tilde{\Omega}, \Omega}
  &= -\log(\tilde{\rho}) \otimes \one
    + \one \otimes \log(\rho)
\end{align*}
acting on $\Hcal_1\otimes\Hcal_1^*$ and hence
\begin{align*}
    H(\Omega,\tilde{\Omega})
  &= \langle \Omega, K_{\tilde{\Omega}, \Omega} \Omega \rangle
\nexteq
  &= \tr_{\Hcal_1}\Bigl(
      \sqrt{\rho}
      \bigl[
        -\log(\tilde{\rho}) \sqrt{\rho}
        + \sqrt{\rho} \log(\rho)
      \bigr]
    \Bigr)
\nexteq
  &= \tr_{\Hcal_1}\Bigl( \rho \bigl[ \log(\rho) - \log(\tilde{\rho}) \bigr] \Bigr)
  \eqend{.}
\end{align*}
\end{example}

Now suppose that $\hat{\Omega}$ is another standard unit vector for $\Rcal$ which gives the same expectation values as $\Omega$. Then there is a unitary $u'\in\Rcal'$ given by
\begin{align}
u'a\Omega&\coloneq a\hat{\Omega}\label{def:u'}
\end{align}
for all $a\in\Rcal$. 
Because both vectors define the same state on $\Rcal$, one expects that
$H(\hat{\Omega},\tilde{\Omega})=H(\Omega,\tilde{\Omega})$ and $H(\tilde{\Omega},\hat{\Omega})=H(\tilde{\Omega},\Omega)$. We will now explain why this does indeed work out, taking a slightly more general perspective.
\begin{lemma}\label{lem:uu'Modular}
Suppose that $u,v\in\Rcal$ 
and $u',v'\in\Rcal'$ are unitaries. Then
\begin{subequations}
\begin{align}
    J_{u u' \tilde{\Omega}, v v' \Omega} 
  &= u' v J_{\tilde{\Omega}, \Omega} u^* v'^* \notag\\
    K_{u u' \tilde{\Omega}, v v' \Omega} 
  &= v' u K_{\tilde{\Omega}, \Omega} u^* v'^*
\label{eqn:uu'vv'K}\\
H(vv'\Omega,uu'\tilde{\Omega})
&=\langle u^*v\Omega,K_{\tilde{\Omega},\Omega}u^*v\Omega\rangle\eqend{.}\label{eqn:REuu'}
\end{align}
\end{subequations}
\end{lemma}
\proof
We will first prove the first two identities when $v=v'=\one$. Because $uu'\Rcal(uu')^*=u\Rcal u^*=\Rcal$ we see from Lemma \ref{lem:UModular} that $uu'\tilde{\Omega}$ is again standard for $\Rcal$. For all $a\in\Rcal$ we have
\begin{align}
S_{uu'\tilde{\Omega},\Omega}a\Omega&=a^*uu'\tilde{\Omega}\notag\\
&=u'(u^*a)^*\tilde{\Omega}\notag\\
&=u'S_{\tilde{\Omega},\Omega}u^*a\Omega\eqend{.}\notag
\end{align}
As $\Rcal\Omega$ is a core for both
$S_{uu'\tilde{\Omega},\Omega}$ and $u'S_{\tilde{\Omega},\Omega}u^*$ we find $S_{uu'\tilde{\Omega},\Omega}=u'S_{\tilde{\Omega},\Omega}u^*$. Using $\Delta_{\tilde{\Omega},\Omega}=S_{\tilde{\Omega},\Omega}^*S_{\tilde{\Omega},\Omega}$,
$J_{\tilde{\Omega},\Omega}=S_{\tilde{\Omega},\Omega}\Delta_{\tilde{\Omega},\Omega}^{-\frac12}$, and the functional calculus, this yields the first two identities when $v=v'=\one$. The general case then follows from Lemma \ref{lem:UModular} for the unitary $vv'$ and a change of notation. For the final identity we compute
\begin{align*}
    H(v v' \Omega, u u' \tilde{\Omega})
  &= \innerProd{v v' \Omega}{v' u K_{\tilde{\Omega}, \Omega}u^* v'^* v v' \Omega}
\nexteq
  &= \innerProd{u^* v \Omega}{K_{\tilde{\Omega}, \Omega} u^* v \Omega}
  \eqend{.}
  \qquad\qed
\end{align*}
Note that $u',v'$ cancel out in \eqref{eqn:REuu'} as anticipated above, so in particular 
$H(v'\Omega,u'\tilde{\Omega})=H(\Omega,\tilde{\Omega})$. If $\Rcal$ is a type III$_1$ von Neumann factor, a theorem of Connes and St{\o}rmer~\cite{Connes1978} tells us that the set of states
$\{\tilde{\Omega}=u'u\hat{\Omega}\mid u\in\Rcal, u'\in\Rcal'\}$ is dense in the unit ball of $\Hcal$ when $\Omega$ is standard for $\Rcal$. In the case where $\tilde{\Omega}=\Omega$ we find for the relative entropy
\begin{align}
\label{eqn:REuu'2}
    H(vv'\Omega,uu'\Omega)
  &= \innerProd{u^*v \Omega}{K_{\Omega} u^*v \Omega}
  \eqend{.}
\end{align}

\begin{remark}
In our applications to QFT, we may start with algebraic states $\omega,\tilde{\omega}$ on a $C^{\star}$-algebra $\Acal$ and write $H(\tilde{\omega},\omega)$ instead of $H(\tilde{\Omega},\Omega)$, where we assume that both states can be represented by standard vectors $\Omega$, resp.~$\tilde{\Omega}$, in the GNS-representation space of $\omega$ and we take $\Rcal=\pi_{\omega}(\Acal)''$. We can see from \eqref{eqn:REuu'} that the choice of the representing vectors is irrelevant for the computation of the relative entropy, because the unitary $u'\in\Rcal'$ defined as in \eqref{def:u'} cancels out.    
\end{remark}

In the following lemma we consider two von Neumann algebras, $\Rcal_1$ and $\Rcal_2$, and we distinguish the corresponding modular data by the same label.
\begin{lemma}\label{lem:ModularEstimate}
Let $\Rcal_1\subset\Rcal_2$ be an inclusion of von Neumann algebras acting on $\Hcal$ and $\Omega$, $\tilde{\Omega}$ standard vectors for both $\Rcal_i$. Then
\begin{subequations}
\begin{align}
K^{(1)}_{\tilde{\Omega},\Omega}&\le
K^{(2)}_{\tilde{\Omega},\Omega}\label{eqn:RMHestimate}\\
H^{(1)}(\Omega,\tilde{\Omega})&\le
H^{(2)}(\Omega,\tilde{\Omega})\eqend{.}\label{eqn:REestimate}
\end{align}
\end{subequations}
\end{lemma}
\begin{proof}
The Tomita operator $S^{(2)}_{\tilde{\Omega},\Omega}$ for $\Rcal_2$ extends the Tomita operator $S^{(1)}_{\tilde{\Omega},\Omega}$ for $\Rcal_1$, so for any $a\in\Rcal_1$ we have 
\begin{align*}
    \left\|
      \left( \Delta^{(2)}_{\tilde{\Omega}, \Omega} \right)^{\frac12}
      a \Omega
    \right\|
  &= \left\| S^{(2)}_{\tilde{\Omega}, \Omega} a \Omega \right\|
  = \| a^* \tilde{\Omega} \|
  = \left\|
      \left( \Delta^{(1)}_{\tilde{\Omega}, \Omega} \right)^{\frac12}
      a \Omega
    \right\|
  \eqend{.}
\end{align*}
Now, $\bigl( \Delta^{(1)}_{\tilde{\Omega}, \Omega} \bigr)^{1/2}$ has a dense range on the core $\Rcal_1 \Omega$, so on this range we may define the operator
\begin{align*}
    C
  &\coloneq \left( \Delta^{(2)}_{\tilde{\Omega}, \Omega} \right)^{\frac12}
    \left( \Delta^{(1)}_{\tilde{\Omega}, \Omega} \right)^{-\frac12}
  \eqend{,}
\end{align*}
which is an isometry by the above computation, $C^*C=\one$. In particular, $\|C^*\| = \|C\| = 1$, so we find $C C^* \leq \one$, where
$C^* a \Omega = \bigl( \Delta^{(1)}_{\tilde{\Omega}, \Omega} \bigr)^{-1/2} \bigl( \Delta^{(2)}_{\tilde{\Omega}, \Omega} \bigr)^{1/2} a \Omega$ for all $a \in \Rcal_2$.
This entails
\begin{align*}
    0
  &< \left( \Delta^{(1)}_{\tilde{\Omega}, \Omega} \right)^{-1}
  \leq \left( \Delta^{(2)}_{\tilde{\Omega}, \Omega} \right)^{-1}
\end{align*}
on the domain of $\bigl( \Delta^{(2)}_{\tilde{\Omega}, \Omega} \bigr)^{-1/2}$.
Since the $\log$ function is operator monotone (cf. \cite[Thm.~4.1]{Simon2019MatrixMonotone}), the result follows from the definitions \eqref{def:relK} and \eqref{eqn:defRE}.
\end{proof}

We will use the inequalities \eqref{eqn:RMHestimate} and \eqref{eqn:REestimate} in combination with \eqref{eqn:uu'vv'K}, \eqref{eqn:REuu'} and \eqref{eqn:REuu'2} by means of the following basic estimates.
\begin{theorem}\label{Prop:GeneralEstimate}
Let $\Rcal_1\subset\Rcal_2$ be an inclusion of von Neumann algebras acting on a Hilbert space $\Hcal$ and let $u,v\in\Rcal_2$ and $u',v'\in\Rcal_1'$ be unitaries.
\begin{enumerate}[label=(\roman*)]
\item If both vectors $u\Omega$ and $v\Omega$ are standard for both $\Rcal_i$, then
\begin{subequations}
\begin{align}
K^{(1)}_{u'u\Omega,v'v\Omega}&\le v'u K^{(2)}_{\Omega}u^*v'^*
\label{eqn:RMHestimate2}\\
H^{(1)}(v'v\Omega,u'u\Omega)
&\le \innerProd{u^*v \Omega}{K_{\Omega}^{(2)} u^*v \Omega}\eqend{.}\label{eqn:REestimate2}
\end{align}
\end{subequations}
\item If both vectors $u'\Omega$ and $v'\Omega$ are standard for both $\Rcal_i$, then
\begin{subequations}
\begin{align}
K^{(2)}_{uu'\Omega,vv'\Omega}&\ge
uv'K^{(1)}_{\Omega}v'^*u^*
\label{eqn:RMHestimate3}\\
H^{(2)}(vv'\Omega,uu'\Omega)
&\ge \innerProd{v'^*u^*vv' \Omega}{K_{\Omega}^{(1)} v'^*u^*vv'\Omega}\eqend{.}\label{eqn:REestimate3}
\end{align}
\end{subequations}
\end{enumerate}
\end{theorem}
Beware of the different ordering of $u$ and $u'$, resp. $v$ and $v'$, in the two statements of this theorem.
\proof
Even if $u$ and $v$ are not in $\Rcal_1$ and if they do not commute with $u',v'$, we can estimate 
\begin{align}
K^{(1)}_{u'u\Omega,v'v\Omega}&=v'K^{(1)}_{u\Omega,v\Omega}v'^*\notag\\
&\le v'K^{(2)}_{u\Omega,v\Omega}v'^*\notag\\
&=
v'uK^{(2)}_{\Omega}u^*v'^*\eqend{,}\notag
\end{align}
where we used \eqref{eqn:uu'vv'K} for $u',v'$ in the first line and for $u,v$ in the last line, and \eqref{eqn:RMHestimate} for the inequality. Analogously,
\begin{align*}
    H^{(1)}(v' v \Omega, u' u \Omega)
  &= H^{(1)}(v\Omega,u\Omega)
\nexteq
  &\le H^{(2)}(v\Omega,u\Omega)
\nexteq
  &= \innerProd{u^*v \Omega}{K_{\Omega}^{(2)} u^*v \Omega}
  \eqend{,}
\end{align*}
where we first used the identity \eqref{eqn:REuu'} to discard $u',v'$, then applied the estimate \eqref{eqn:REestimate} and finally used the identity \eqref{eqn:REuu'2}, but now for the larger algebra $\Rcal_2$, which does contain $u$ and $v$.
Similarly we can estimate
\begin{align*}
    K^{(2)}_{u u' \Omega, v v' \Omega}
  &= u K^{(2)}_{u' \Omega, v' \Omega} u^*
\nexteq
  &\geq u K^{(1)}_{u' \Omega, v' \Omega} u^*
\nexteq
  &= u v' K^{(1)}_{\Omega} v'^* u^*
\end{align*}
and
\begin{align*}
    H^{(2)}(v v' \Omega, u u' \Omega)
  &= \innerProd{u^* v v' \Omega}{
      K_{u' \Omega, v' \Omega}^{(2)} u^* v v' \Omega
    }
\nexteq
  &\geq \innerProd{u^* v v' \Omega}{
      K_{u' \Omega, v' \Omega}^{(1)} u^* v v' \Omega
    }
\nexteq
  &= \innerProd{v'^* u^* v v' \Omega}{K_{\Omega}^{(1)}v'^* u^* v v' \Omega}
  \eqend{.}
  \qquad\qed
\end{align*}
\begin{remark}
By similar arguments, using an exponentiated version of \eqref{eqn:uu'vv'K} in Lemma \ref{lem:uu'Modular} 
and $\Delta^{(1)}_{\tilde{\Omega},\Omega}\ge\Delta^{(2)}_{\tilde{\Omega},\Omega}$ instead of Lemma \ref{lem:ModularEstimate} one can also show that under the assumptions of Theorem \ref{Prop:GeneralEstimate} we have (i) $\Delta^{(1)}_{u'u\Omega,v'v\Omega}\ge v'u\Delta^{(2)}_{\Omega}u^*v'^*$ and (ii) $\Delta^{(2)}_{uu'\Omega,vv'\Omega}\le uv'\Delta^{(1)}_{\Omega}v'^*u^*$.\footnote{K.S. thanks Stefan Hollands for enquiring about these analogs to Theorem \ref{Prop:GeneralEstimate}.}
\end{remark}
Note that the assumptions of Theorem \ref{Prop:GeneralEstimate} only depend on the states on the algebra $\Rcal_2$, not on the choice of vectors that implement these states.
Indeed, if $\tilde{\Omega}$ defines the same state as $v' v \Omega$ on
$\Rcal_2$, then $\tilde{\Omega} = w' v' v \Omega$ for some unitary $w' \in \Rcal_2'$ by \eqref{def:u'}, so the assumption of item (i) still holds with $w' v'$ instead of $v'$.
A similar argument holds for the vectors $u' u \Omega$, $v v' \Omega$ and $u u' \Omega$.
In the estimates for the relative entropies, the unitaries from $\Rcal_2'$ cancel out.

\subsection{Estimates on relative modular Hamiltonians in quantum field theory}\label{ssec:applicationinQFT}
Let $M=(\Mcal,g)$ be a smooth globally hyperbolic Lorentzian manifold. To describe a QFT in $M$ we assign to each causally convex open subset $V\subset M$ a $C^{\star}$-algebra $\Acal_V$, which we think of as being generated by smearing any quantum fields of the theory with smooth test functions (or tensors) compactly supported in $V$. We will make the following standard assumption:
\begin{align}
V_1\subset V_2\quad&
\Rightarrow\quad
\Acal_{V_1}\subset\Acal_{V_2}\eqend{.}&(\text{\emph{isotony}})\notag
\end{align}
Any algebraic state $\omega$ on $\Acal_M$ restricts to a state on $\Acal_V$ for any causally convex open region $V\subset M$ and we denote by $\Rcal^{(V)}_{\omega}$ the von Neumann algebra $\pi_{\omega}(\Acal_V)''$ generated from $\Acal_V$ in the GNS-representation of $\omega|_{\Acal_V}$.

If $V$ is bounded, physically reasonable states on $\Acal_M$ are expected to have quasi-equivalent restrictions to $\Acal_V$. This means that every such state can be expressed by a density matrix in the GNS-representation of any other such state. This explains why the choice of Hilbert space representation is irrelevant for the description of local phenomena in QFT, cf. \cite{witten2022does,haag2012local}.

In many cases the GNS-vector $\Omega_{\omega}$ is a standard vector for 
$\Rcal^{(V)}_{\omega}$. This occurs e.g.~when $\omega$ is a vacuum vector in Minkowski space and $V$ is a bounded open region \cite{ReehSchlieder:1961}. For more details on this property we refer to Section II of the nice review by Witten \cite{Witten:RevModPhys.90.045003}. 

Given any two quasi-equivalent states $\omega$, $\tilde{\omega}$ on $\Acal_V$ whose GNS-vectors are both standard, the Unitary Implementation Theorem applies, so we can represent both states by unit vectors $\Omega$ and $\tilde{\Omega}$ in a single (GNS-)Hilbert space representation. Moreover, we can introduce the relative modular theory of Section \ref{sec:ModularGeneralities}. Let us fix a standard unit vector $\hat{\Omega}$ for $\Rcal=\Rcal_{\omega}^{(V)}$ for reference purposes. If we can write $\Omega=vv'\hat{\Omega}$ and $\tilde{\Omega}=uu'\hat{\Omega}$ for some unitaries $u,v\in\Rcal$ and $u',v'\in\Rcal'$, then we can obtain results about the relative modular Hamiltonian and the relative entropy directly from Lemma \ref{lem:uu'Modular}. This case already covers a lot of states, due to the Connes-St{\o}rmer Theorem \cite{Connes1978}.

More generally, if we have three causally convex regions $V_1\subset V_2\subset V_3$ in $M$ and if $\hat{\Omega}$ is a reference state that is standard for all three algebras $\Rcal_i \coloneq \pi_{\hat{\omega}}(\Acal_{V_i})''$, $i=1,2,3$, then we may try to apply the results of Theorem \ref{Prop:GeneralEstimate} to estimate the relative modular Hamiltonian $K^{(2)}_{\tilde{\Omega},\Omega}$ and the relative entropy
$H^{(2)}(\Omega,\tilde{\Omega})$ from below, in terms of $K^{(1)}_{\hat{\Omega}}$, and from above, in terms of  $K^{(3)}_{\hat{\Omega}}$. To apply the theorem, however, we need to be able to express the states $\tilde{\Omega}$ and $\Omega$ in terms of the reference state $\hat{\Omega}$ and suitably localised unitaries. To examine what this implies, let us analyse item (i) of Theorem~\ref{Prop:GeneralEstimate}, considering only the two regions \(V_2\subset V_3\) (with appropriate relabelling), cf. Figure \ref{subfig:sorkin2}. (A similar argument could be made for item (ii) with regions \(V_1\subset V_2\).)

Item (i) of the Theorem \ref{Prop:GeneralEstimate} requires $\Omega=v'v\hat{\Omega}$, with \(v'\in\Rcal_2'\), \(v\in\Rcal_3\) and similarly for $\tilde{\Omega}$. Note that these conditions are independent of the choice of vectors, as long as they define the same state on $\Rcal_3$. Indeed, if $\Omega'$, resp.~$\hat{\Omega}'$, define the same states on $\Rcal_3$ as $\Omega$, resp.~$\hat{\Omega}$, then we must have $\Omega'=w\Omega$, $\hat{\Omega}'=\hat{w}\hat{\Omega}$ for some unitaries $w,\hat{w}\in\Rcal_3'$. The assumption of item (i) is then still verified with
$wv'\hat{w}^*$ instead of $v'$.

\medskip

\begin{figure}[htbp]
  \centering
    \begin{subfigure}{0.45\textwidth}
        \centering
        \includegraphics{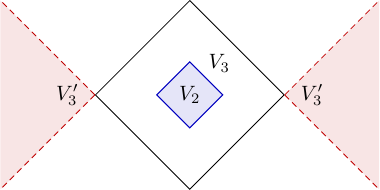}
        \caption{Regions in Theorem \ref{Prop:GeneralEstimate}(i).}
        \label{subfig:sorkin2}
    \end{subfigure}
    \hfill
    \begin{subfigure}{0.45\textwidth}
        \centering
        \includegraphics{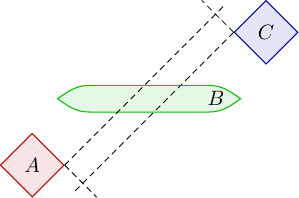}
        \caption{Sorkin's impossible measurement setup.}
        \label{subfig:sorkin1}
    \end{subfigure}
      \caption{A comparison of the setup described in~\cite{Sorkin:1993} and Theorem \ref{Prop:GeneralEstimate}(i).}
    \label{fig:sorkin}
\end{figure}

\medskip

To simplify the discussion, take the case where \(\Omega=\hat{\Omega}\), i.e. one of the two states is our reference state. Then, in this setting, the theorem requires $\tilde{\Omega}=u'u\Omega$, with \(u'\in\Rcal_2'\), \(u\in\Rcal_3\). It is then interesting to note that the applicability of Theorem~\ref{Prop:GeneralEstimate} is related to Sorkin's paradox. In~\cite{Sorkin:1993}, Sorkin points out that it is possible to construct examples of superluminal signalling, so-called \enquote{impossible measurements}. More recently,~\cite{borsten2021impossible} and~\cite{jubb2022causal} have formulated this discussion in terms of conditions on unitary operations to be non-signalling. The typical setting of these impossible measurements considers three observers, Alice, Bob and Charlie, located in corresponding regions \(A,B,C\), such that \(A\) and \(C\) are spacelike separated, so Alice's actions cannot influence Charlie. However, \(B\) may lie partially in the future of \(A\) and partially in the past of \(C\), so Alice's actions can influence Bob and Bob's actions can influence Charlie, cf. Figure \ref{subfig:sorkin1}. An impossible measurement is a unitary operation performed by Bob in \(B\) that allows Charlie to determine whether Alice has performed some unitary operation on the shared quantum system or not. The existence of such impossible measurements is paradoxical, because it apparently violates the causality between \(A\) and \(C\).

In the context of QFT, we may formalise this situation in the following way. Let \(\Rcal_A\) and \( \Rcal_C\) be the local algebras relative to regions \(A\) and \(C\) acting on a Hilbert space \(\Hcal\), so \(\left[\Rcal_A,\Rcal_C\right]=0\). Then, Bob's measurement in \(B\) is given by the adjoint action of some unitary \(w\) on \(\Hcal\) and
\begin{eqnarray}
w\ \text{does not allow signalling from}\ A\ \text{to}\ C&\text{iff}\quad w \Rcal_Aw^*\subset\Rcal_C'\eqend{,}\notag\\
w\ \text{does not allow signalling from}\ C\ \text{to}\ A&\text{iff}\quad w\Rcal_C w^*\subset\Rcal_A'\eqend{,}\notag
\end{eqnarray}
cf.~\cite[Sec.~III]{borsten2021impossible}.

As an example we consider unitaries of the simple product form \(w=u'u\), with \(u'\in\Rcal_C'\), \(u\in \Rcal_A'\). Such operators are clearly non-signalling, as
\begin{equation*}
w\Rcal_Aw^*=u'u\Rcal_Au^*u'^*=u'\Rcal_Au'^*\subset\Rcal_C'\eqend{.}
\end{equation*}
We may now relate the discussion back to the unitaries of Theorem~\ref{Prop:GeneralEstimate}, considering the regions \(V_2\subset V_3\) and identifying \(A=V_3'\) and \(C=V_2\) as shown in Figure~\ref{fig:sorkin}, so \(u\in\Rcal_A'=\Rcal_3\) and \(u'\in\Rcal_C'=\Rcal_2'\). In this way, we are comparing signalling in the direction \(A\) to \(C\) with Theorem~\ref{Prop:GeneralEstimate}(i) (similarly, (ii) could be compared with signalling from \(C\) to \(A\) considering regions \(V_1\subset V_2\)).

Now, one may ask how many pairs of states can be covered by our scheme. As mentioned before, Connes and
St{\o}rmer~\cite{Connes1978} tell us that if $\Rcal_2$ is a type III$_1$ von Neumann factor, the set of states
$\{\tilde{\Omega}=u'u\hat{\Omega}\mid u\in\Rcal_2, u'\in\Rcal_2'\}$ is dense in the unit ball of $\Hcal$. Because
$u'u\hat{\Omega}=uu'\hat{\Omega}$ we see that this set of states satisfies the assumptions of part (i) of Theorem~\ref{Prop:GeneralEstimate} w.r.t.~$\Rcal_2$ and $\Rcal_3$ and of part (ii) w.r.t.~$\Rcal_1$ and $\Rcal_2$. We can therefore certainly use any $\Omega,\tilde{\Omega}$ from this dense set.

On the other hand, it follows from the previous paragraphs that  for $\Omega=\hat{\Omega}$ our upper estimate only applies if there is a unitary $w$ such that $\tilde{\Omega}=w\Omega$ and $w$ does not allow signalling from $V_3'$ to $V_2$. (If we allow general unitaries on $\Hcal$, this may not preclude the existence of another unitary $w'$ such that $\tilde{\Omega}=w'\Omega$ and $w'$ does allow signalling.) One may now ask whether the existence of such a non-signalling unitary $w$ satisfying
$\tilde{\Omega}=w\Omega$ is also sufficient to ensure that our upper estimate applies. In complete analogy, for our lower bound to apply it is necessary that there exists a unitary $w$ such that $\tilde{\Omega}=w\Omega$ and $w$ does not allow signalling from $V_1$ to $V_2'$ and one may wonder whether the existence of such a $w$ is also sufficient for the lower bound to apply.

For type I\(_\infty\) factors the following (counter-)example shows that this is not the case.
\begin{example}[Type I\(_\infty\) factors]\label{ex:nonproductsignallingunitary} 
Let \(\Hcal\) be an infinite dimensional Hilbert space, \(n\in\Nbb\) and \(\{u_i\}_{i=1}^n\) operators on \(\Hcal\) satisfying
\begin{align}
\label{eq:CuntzConditions}
    u_i u_j^*
  &= \delta_{ij} \one
  \eqend{,}
&
    \sum_i u_i^* u_i
  &= \one
  \eqend{.}
\end{align}
Then, the operators \(u_i\) generate a Cuntz algebra. We note that a von Neumann algebra \(\Rcal\) acting on \(\Hcal\) contains operators \(u_i\in\Rcal\) with these properties if it is a von Neumann factor of type I\(_{\infty}\) or of type III, cf. \cite[Def. 6.5.1 and Lemma 6.3.3]{KadisonRingroseII}.

Now consider the type I\(_{\infty}\) von Neumann factors \(\Rcal_2=\Rcal_3=\Cbb\one\otimes \Bcal(\Hcal)\) on
\(\Hcal\otimes\Hcal\) and operators \(u_i\in \Rcal_3\) and \(u'_i\in\Rcal_2'=\Bcal(\Hcal)\otimes \Cbb\one\) such that \(u_i\) and
\((u'_i)^*\) satisfy the relations \eqref{eq:CuntzConditions}. Then \(w \coloneq \sum_{i=1}^n u'_iu_i\) is a unitary which is non-signalling between \(\Rcal_3'\) and \(\Rcal_2\), because
\begin{equation}\label{eq:nonsignallingsums}
w \Rcal_3' w^*=\sum_{i,j=1}^n u_i'u_i\Rcal_3'u_j^*u_j'^*=\sum_{i,j=1}^nu_i'u_iu_j^*\Rcal_3'u_j'^*=\sum_{i=1}^n u_i'\Rcal_3'u_i'^*\subset\Rcal_2'\eqend{.}
\end{equation}

We may write \(u_i=\one\otimes v_i\) and \(u_i'=v_i'\otimes\one\), where \(v_i\) and \((v'_i)^*\) also satisfy the relations \eqref{eq:CuntzConditions}. Let the subspaces \(\Hcal'_i\) be the ranges of the projections \(v_i'v_i'^*\), so
\(\Hcal=\bigoplus_{i=1}^n\Hcal'_i\) and similarly for \(\Hcal_i\) with \(v_i^*v_i\). For \(n>1\), we will choose a suitable standard vector
\(\hat{\Omega}\) for \(\Rcal_2=\Rcal_3\) as follows. We fix two orthonormal vectors \(\chi_1,\chi_2\in\Hcal\) and we define the unit vector \(\psi_1 \coloneq \frac{1}{\sqrt{2}}(v_1^*\chi_1+v_2^*\chi_2)\). We extend the vector \(\psi_1\) to an orthonormal basis \(\{\psi_i\}\) 
of \(\Hcal\) and we choose another orthonormal basis \(\{\psi'_i\}\). For any \(\varepsilon\in(0,1)\) we then choose a sequence of positive numbers \(\{c_i\}\) with \(c_1=1-\varepsilon\) and we set \(\hat{\Omega} \coloneq \sum_{i=1}^{\infty}c_i\psi_i'\otimes\psi_i\). This is indeed a standard vector, because \(c_i>0\) for all \(i\). Now suppose that there exist two unitaries \(u,u'\) on \(\Hcal\) such that
\(w\hat{\Omega}=u'\otimes u\hat{\Omega}\). We will derive a contradiction, showing that no such \(u,u'\) exist.

Note that \(\|\hat{\Omega}-c_1\psi_1'\otimes\psi_1\|^2=1-c_1^2=2\varepsilon-\varepsilon^2\), so
\(\|\hat{\Omega}-c_1\psi_1'\otimes\psi_1\|\le\sqrt{2\varepsilon}\). Because \(w\) and \(u'\otimes u\) are unitaries, we similarly have
\(\|w(\hat{\Omega}-c_1\psi_1'\otimes\psi_1)\|\le\sqrt{2\varepsilon}\) and
\(\|u'\otimes u(\hat{\Omega}-c_1\psi_1'\otimes\psi_1)\|\le\sqrt{2\varepsilon}\) and hence
\begin{align}
\|(u'\otimes u)\hat{\Omega}-w\hat{\Omega}\|&\ge (1-\varepsilon)\|u'\psi_1'\otimes u\psi_1-w\psi_1'\otimes\psi_1\|-2\sqrt{2\varepsilon}\eqend{,}\notag
\end{align}
where we used \(c_1=1-\varepsilon\). Now, \(w=\sum_{i=1}^nv_i'\otimes v_i\) and 
\(v_i\psi_1=\frac{1}{\sqrt{2}}\chi_i\) 
with 
\(\chi_i=0\) for \(i>2\) imply
\begin{align}
w\psi'_1\otimes\psi_1&=\sum_{i=1}^n v_i'\psi'_1\otimes v_i\psi_1=\frac{1}{\sqrt{2}}\sum_{i=1}^2 v_i'\psi'_1\otimes \chi_i\eqend{.}\notag
\end{align}
With this and the Cauchy-Schwarz inequality we estimate
\begin{align*}
    \left\| u' \psi_1' \otimes u \psi_1 - w \psi_1' \otimes \psi_1 \right\|^2
  &= 2 - \sqrt{2} \sum_{i = 1}^2 \Re\langle u' \psi_1', v_i' \psi'_1 \rangle \; \langle u \psi_1, \chi_i \rangle
\nexteq
  &\geq 2 - \sqrt{2} \sum_{i = 1}^2 |\langle u'\psi_1', v_i' \psi'_1 \rangle| \; |\langle u \psi_1, \chi_i \rangle|
\nexteq
  &\geq 2 - \sqrt{2} \sqrt{\sum_{i = 1}^2 |\langle u' \psi_1', v_i' \psi'_1 \rangle|^2} \sqrt{\sum_{i = 1}^2 |\langle u \psi_1, \chi_i \rangle|^2}
\nexteq
  &\geq 2 - \sqrt{2} \sqrt{\| u' \psi_1' \|^2} \sqrt{\| u \psi_1 \|^2}
\nexteq
  &= 2 - \sqrt{2}
  \eqend{,}
\end{align*}
where we used the fact that the vectors \(v_i'\psi'_1\) are orthonormal. It then follows that
\begin{align*}
    \left\| (u' \otimes u) \hat{\Omega} - w \hat{\Omega} \right\|
  &\geq (1 - \varepsilon) \sqrt{2 - \sqrt{2}} - 2\sqrt{2 \varepsilon}
\end{align*}
which is strictly positive when $\varepsilon$ is sufficiently small.
\end{example}

Note that the unitaries \(w\) described in this example are thus non-signalling but not of the simple product form \(w=u'u\) of Theorem~\ref{Prop:GeneralEstimate} when \(n>1\), for otherwise \(w\psi'_1\otimes\psi_1\) would also be a simple tensor product, which is not the case. In the QFT case, however, if \(\Rcal_2'\cap \Rcal_3\) is a type III\(_1\) factor, then \(w=u'u\). To see this, we note that we can choose operators \(c_1,\ldots,c_n\in\Rcal_2'\cap \Rcal_3\) generating a Cuntz algebra, so \(c_i^*c_j=\delta_{ij}\one\) and \(\sum_i c_ic_i^*=\one\). Then we can define \(u' \coloneq \sum_i u_i'c_i^*\) and \(u \coloneq \sum_i c_i u_i\), which will be unitary by the relations~\eqref{eq:CuntzConditions}, and we get that
\begin{equation*}
    u'u=\sum_{i,j=1}^n u_i' c_i^*c_ju_j =w \eqend{.}\notag
\end{equation*}
The reason why this argument breaks down in the example is that there we had \(\Rcal_2'\cap\Rcal_3=\Cbb\one\otimes\one\), which does not contain operators \(c_i\) satisfying~\eqref{eq:CuntzConditions}.

\section{Generalities on coherent states}\label{sec:coherentstates}

In this section we consider coherent states in general, bosonic CCR systems. Before we apply our general estimates to these states, we first review the fact that one can often compute relative modular Hamiltonians explicitly due to some remarkable simplifications. Moreover, the relative entropy can always be determined, but it might be infinite. Let us begin by 
fixing our notations following \cite[Sec.~5.2.1]{brattelirobinson1997} and \cite[Sec.~X.7]{ReedSimonII}.

\subsection{Fock space preliminaries}\label{ssec:fockspacebasics}

Let $\Hcal$ be a complex Hilbert space and $\Fcal_+(\Hcal)=\bigoplus_{n=0}^{\infty}\Hcal^{\otimes_sn}$ its bosonic Fock space
with Fock vector $\Omega_{\mathrm{F}} \in \Hcal^{\otimes_s 0} = \Cbb$.
For any unitary $u$ on $\Hcal$ we denote the second quantised unitary on $\Fcal_+(\Hcal)$ by $\Gamma(u)$, so that $\Gamma(u)^* = \Gamma(u^*)$.
For a self-adjoint operator $h$ on $\Hcal$ we define $\d\Gamma(h)$ implicitly by $e^{\i t \d\Gamma(h)} = \Gamma(e^{\i t h})$ for all $t \in \Reals$.

For any $\chi \in \Hcal$ we let $a^*(\chi)$ and $a(\chi)$ denote the usual creation and annihilation operators, satisfying 
$a(\chi) \Omega_{\mathrm{F}} = 0$ and the canonical commutation relations
\begin{align*}
    \comm{a(\chi)}{a^*(\xi)}
  &= \langle \chi, \xi \rangle \one
\nexteq
    \comm{a(\chi)}{a(\xi)}
  &= 0
\end{align*}
for all $\chi,\xi\in\Hcal$. Note that the operators $a^*(\chi)$ and $a(\chi)$ are closed, unbounded if $\chi\not=0$ and each other's adjoints. We define the Segal field operator by
\begin{align*}
    \phi_{\textup{S}}(\chi)
  &\coloneq \frac{1}{\sqrt{2}}
    \bigl( a^*(\chi) + a(\chi) \bigr)
\end{align*}
for any $\chi\in\Hcal$. This operator is self-adjoint and unbounded if $\chi\not=0$.

Weyl operators are unitaries defined by $W(\chi) \coloneq \e^{\i \phi_{\textup{S}}(\chi)}$ for $\chi\in\Hcal$ and they satisfy the Weyl relations
\begin{align}
\label{eqn:Weylrelations}
    W(\chi) W(\xi)
  &= \e^{-\frac{\i}{2} \Im\langle\chi,\xi\rangle} W(\chi + \xi)
\end{align}
for all $\chi,\xi\in\Hcal$. For any closed real linear subspace $H\subset\Hcal$ we consider the von Neumann algebra $\Rcal(H) \coloneq \{W(\chi)\mid \chi\in H\}''$ generated by the Weyl operators with $\chi\in H$. It is known that $\Omega_{\mathrm{F}}$ is a cyclic, resp.~separating, vector for
$\Rcal(H)$ iff $\overline{H + \i H} = \Hcal$, resp.~$H \cap \i H = \{ 0 \}$, cf. \cite[Thm.~3.12]{araki1982quasi}. To have a standard vector we therefore need to assume both these properties and in that case $H$ is called a standard subspace. If $H\subset\Hcal$ is standard, then it is known that the modular Hamiltonian of \(\Rcal(H)\) for \(\Omega_{\mathrm{F}}\) is of second quantised form,
\begin{align*}
    K_{\Omega_{\mathrm{F}}}
  &= \d\Gamma(K_H)
  \eqend{,}
\end{align*}
where $K_H$ is a self-adjoint operator on $\Hcal$. Indeed, modular theory can be developed for standard subspaces \cite{longo2008real}, where we have $K_H=-\log(\Delta_H)$ and $\Delta_H$ appears in the polar decomposition $S_H = J_H \sqrt{\Delta_H}$ of the closed antilinear operator $S_H : H + \i H \to H + \i H$ defined by $S_H(\chi + \i \xi) = \chi - \i \xi$ for all $\chi, \xi \in H$.

For later convenience we mention here the following technical lemma.
\begin{lemma}
\label{lem:strongderivative}
Let $I\subset\Rbb$ be an open interval containing $0$ and $h:I\to\Hcal$ a continuous $\Hcal$-valued function with $h(0)=0$. If
$h'(0)$ exists and $\psi\in\Fcal_+(\Hcal)$ is a finite particle vector, then
\begin{align*}
    \left. \deriv{}{t} W(h(t)) \psi \right|_{t = 0}
  &= \i \phi_{\textup{S}} (h'(0)) \psi
  \eqend{.}
\end{align*}
\end{lemma}
\begin{proof}
The proof uses standard particle number estimates. Suppose $\psi\in\bigoplus_{n=0}^N\Hcal^{\otimes_sn}$ has at most $N$ particles for some $N\in\Zbb_{\ge0}$. For any $\chi\in\Hcal$ we then have $\|a^*(\chi)\psi\|\le\sqrt{N+1}\|\chi\|\cdot\|\psi\|$ and
$\|a(\chi)\psi\|\le\sqrt{N}\|\chi\|\cdot\|\psi\|$ and therefore also $\|\phi_{\textup{S}}(\chi)\psi\|\le\sqrt{2(N+1)}\|\chi\|\cdot\|\psi\|$. Iterating this estimate $n$ times we find $\|\phi_{\textup{S}}(\chi)^n\psi\|\le2^{\frac{n}{2}}\sqrt{\frac{(N+n)!}{N!}}\|\chi\|^n\|\psi\|$ for any $n\in\Nbb$. Using $\sqrt{\frac{(N+n)!}{N!}}\le(N+1)^{\frac{n}{2}}\sqrt{n!}$ this estimate can be brought into the more convenient form
\begin{align}
\label{eqn:absconv}
    \|\phi_{\textup{S}}(\chi)^n\psi\|
  &\leq \bigl( 2 (N + 1) \bigr)^{\frac{n}{2}}
    \|\chi\|^n \sqrt{n!} \|\psi\|
  \eqend{.}
\end{align}
It follows that the power series expansion for $W(\chi)\psi$ converges absolutely, because 
\begin{align*}
    \left\|\sum_{n = 0}^{\infty} \frac{\i^n}{n!} \phi_{\textup{S}}(\chi)^n \psi \right\|
  &\leq \sum_{n = 0}^{\infty}
      \frac{1}{\sqrt{n!}} \bigl( 2 (N + 1) \bigr)^{\frac{n}{2}}
      \|\chi\|^n \|\psi\|
  < \infty
  \eqend{.}
\end{align*}
Similarly, for $t\not=0$ the series
\begin{align*}
    \frac{1}{t} \Bigl[ W\bigl( h(t) \bigr) \psi - \psi \Bigr]
    - \i \phi_{\textup{S}}\bigl( h'(0) \bigr) \psi
  &= \i \phi_{\textup{S}} \left( \frac{1}{t} h(t) - h'(0) \right) \psi
    + \sum_{n = 2}^{\infty} \frac{\i^n}{n!} \frac{1}{t} \phi_{\textup{S}}\bigl( h(t) \bigr)^n \psi
\end{align*}
converges absolutely.
To estimate the terms in the series we use the assumptions on $h(t)$. For any $\varepsilon\in(0,1]$ there is a
$\delta>0$ such that $|t|<\delta$ implies $\left\|\frac{1}{t}h(t)-h'(0)\right\|<\varepsilon\le1$ and hence
$\|h(t)\|<c|t|$ with $c \coloneq 1+\|h'(0)\|$. Using \eqref{eqn:absconv} we then find
\begin{align*}
    \left\|
      \frac{1}{t} \Bigl[ W\bigl( h(t) \bigr) \psi - \psi \Bigr]
      - \i \phi_{\textup{S}}\bigl( h'(0) \bigr) \psi
    \right\|
  &\leq \sqrt{2(N+1)} \varepsilon \|\psi\|
    + \sum_{n = 2}^{\infty} \frac{1}{\sqrt{n!}} \bigl( 2(N + 1) c^2 \bigr)^{\frac{n}{2}}
    |t|^{n - 1} \|\psi\|
\end{align*}
if $|t|<\delta$. The series converges for any $t\in\Rbb$ and we can make it arbitrarily small by choosing $|t|$ sufficiently small. Because $\varepsilon>0$ was arbitrary we find
\begin{align*}
    \lim_{t \to 0} \frac{1}{t} \Bigl[ W\bigl( h(t) \bigr) \psi - \psi\Bigr] -\i \phi_{\textup{S}}\bigl( h'(0) \bigr) \psi
  &= 0
  \eqend{.}
\end{align*}
The result then follows from the observation that $W\bigl( h(0) \bigr) \psi = W(0) \psi = \psi$.
\end{proof}

\subsection{Relative modular Hamiltonians for coherent states}\label{ssec:relmodhamiltonianscoherentstates}

A coherent state, or rather a coherent excitation of $\Omega_{\mathrm{F}}$, is a unit vector of the form $W(\chi)\Omega_{\mathrm{F}}$ for some $\chi\in\Hcal$. Note that the unitary $W(\chi)$ preserves any von Neumann algebra of the form $\Rcal(H)$ with $H\subset\Hcal$ closed, because the Weyl relations imply
\begin{align}
\label{eq:Weylcommutator}
    W(-\chi) W(\xi) W(\chi)
  &= \e^{\i \Im\langle \chi,\xi\rangle}W(\xi)
\end{align}
for any $\xi\in H$. This means that relative modular Hamiltonians exist for any pair of coherent states and also the relative entropy can always be defined, although it may sometimes be infinite.

For many pairs of coherent states the relative modular Hamiltonian can be computed exactly. The following result concerns a slightly more general class of states, but uses the same methods as \cite{Ciolli2020,Longo:2019}.

\begin{proposition}\label{prop:weylactionomodularoperators}
Let $H\subset\Hcal$ be a standard subspace, $\hat{\Omega}\in\Fcal_+(\Hcal)$ a standard vector for $\Rcal(H)$ and
$\Omega \coloneq W(\chi)\hat{\Omega}$, $\tilde{\Omega} \coloneq W(\tilde{\chi})\hat{\Omega}$ for some
$\chi,\tilde{\chi}\in\Hcal$. If there are $h\in H$ and $h'\in H'$ such that $\chi-\tilde{\chi}=h+h'$, then $W(\chi)\Rcal(H)W(-\chi)=\Rcal(H)$ and
\begin{align}
J_{\tilde{\Omega},\Omega}&=W(\tilde{\chi}+h)J_{\hat{\Omega}}W(-\chi+h)\notag\\
K_{\tilde{\Omega},\Omega}&=W(\chi-h)K_{\hat{\Omega}}W(-\chi+h)\eqend{.}\notag
\end{align}
\end{proposition}
\begin{proof}
It follows from \eqref{eq:Weylcommutator} that the adjoint action of $W(\chi)$ preserves a strongly dense sub-algebra of
$\Rcal(H)$ and therefore it preserves $\Rcal(H)$, $W(\chi)\Rcal(H)W(-\chi)=\Rcal(H)$. The same applies to $W(\tilde{\chi})$ and using Lemma \ref{lem:UModular} with 
$u=W(\tilde{\chi})$, the identity $W(-\tilde{\chi}) W(\chi) = e^{\frac{\i}{2}\Im\langle \tilde{\chi},\chi\rangle} W(h) W(h')$, and Lemma \ref{lem:uu'Modular} we find
\begin{align}
    S_{\tilde{\Omega},\Omega}
  &= W(\tilde{\chi}) S_{\hat{\Omega}, W(-\tilde{\chi}) W(\chi)\hat{\Omega}} W(-\tilde{\chi}) \notag\\
  &= e^{\frac{\i}{2}\Im\langle \tilde{\chi},\chi\rangle} W(\tilde{\chi}) W(h) S_{\hat{\Omega}} W(-h') W(-\tilde{\chi}) \notag\\
  &= e^{\frac{\i}{2} \Im\langle \tilde{\chi}, \chi - h\rangle} W(\tilde{\chi} + h) S_{\hat{\Omega}} e^{\frac{\i}{2} \Im\langle \tilde{\chi}, h' \rangle} W(-\tilde{\chi} - h') \notag\\
  &= W(\tilde{\chi} + h) S_{\hat{\Omega}} W(-\chi + h)
  \eqend{,}\notag
\end{align}
where we recall that $\chi-\tilde{\chi}=h+h'$ and $S_{\hat{\Omega}}$ is antilinear. Using the polar decomposition of
$S_{\tilde{\Omega},\Omega}$ and the spectral calculus the results now follow.
\end{proof}

If $\hat{\Omega}=\Omega_{\mathrm{F}}$ is the Fock vector, then the states $\Omega,\tilde{\Omega}$ are coherent. In that case, we can use the Fock space structure to further simplify the formula for the relative modular Hamiltonian in Proposition~\ref{prop:weylactionomodularoperators}.
\begin{corollary}\label{Cor:RelativeModularHamiltonian.CoherentStates}
Under the assumptions of Proposition \ref{prop:weylactionomodularoperators}, if $\hat{\Omega}=\Omega_{\mathrm{F}}$ and $\chi-h$ is in the domain of $K_H$, then
\begin{align*}
    K_{\tilde{\Omega}, \Omega}
  &= K_{\Omega_{\mathrm{F}}}
    + \frac12 \bigl\langle \chi - h, K_H (\chi - h) \bigr\rangle \one
    - \phi_{\textup{S}}\bigl( \i K_H (\chi - h) \bigr)
  \eqend{,}
\end{align*}
when acting on finite particle vectors in the domain of $K_{\Omega_{\mathrm{F}}}$.
\end{corollary}
\begin{proof}
For any unitary $u$ on $\Hcal$ and any $\xi\in\Hcal$, we have $\Gamma(u)a^*(\xi)\Gamma(u)^*=a^*(u\xi)$. Using this identity and its adjoint we find $\Gamma(u)\phi_{\textup{S}}(\xi)\Gamma(u)^*=\phi_{\textup{S}}(u\xi)$ and hence
\begin{align}
\label{eqn:GammaAdjointW}
    \Gamma(u) W(\xi) \Gamma(u)^*
  &= W(u \xi)
  \eqend{.}
\end{align}
By the Weyl relations this yields
\begin{align*}
    W(-\xi) \Gamma(u) W(\xi) \Gamma(u)^*
  &= W(-\xi) W(u\xi)
  = \e^{\frac{\i}{2} \Im\langle\xi, u \xi\rangle}
    W\bigl( (u - \one) \xi \bigr)
  \eqend{.}
\end{align*}
Taking $u = \e^{\i t K}$ for a self-adjoint operator $K$ whose domain contains $\xi$, we may use Lemma \ref{lem:strongderivative} for $h(t) = (\e^{\i t K} - \one) \xi$ to differentiate w.r.t.~$t$ at $t = 0$, which yields
\begin{align}
\label{eqn:WdGamma}
    W(-\xi) \d\Gamma(K) W(\xi) - \d\Gamma(K)
  &= \frac12 \langle\xi, K \xi\rangle \one
    + \phi_{\textup{S}}(\i K \xi)
\end{align}
when acting on any finite particle vector in the domain of $\d\Gamma(K)$. Taking $K=K_H$ and $\xi=h-\chi$ then yields the result.
\end{proof}

Using similar arguments one may compute the relative entropy.
\begin{corollary}\label{Cor:RelativeEntropy.CoherentStates}
Under the assumptions of Proposition \ref{prop:weylactionomodularoperators}, if $\hat{\Omega}=\Omega_{\mathrm{F}}$ and $h$ is in the domain of $K_H$, then
\begin{align*}
    H(\Omega, \tilde{\Omega})
  &= \frac12 \langle h, K_H h\rangle
  \eqend{.}
\end{align*}
\end{corollary}
\begin{proof}
From Proposition \ref{prop:weylactionomodularoperators} and the Weyl relations \eqref{eqn:Weylrelations}, we find
\begin{align*}
    H(\Omega, \tilde{\Omega})
  &= \bigl\langle
      W(h - \chi) \Omega,
      K_{\Omega_{\mathrm{F}}} W(h - \chi) \Omega
    \bigr\rangle
  = \bigl\langle
      \Omega_{\mathrm{F}},
      W(-h) K_{\Omega_{\mathrm{F}}} W(h) \Omega_{\mathrm{F}}
    \bigr\rangle
  \eqend{.}
\end{align*}
Because $h$ is in the domain of $K_H$ and $K_{\Omega_{\mathrm{F}}}\Omega_{\mathrm{F}}=0$, we may apply the identity \eqref{eqn:WdGamma} to find the result. (The fact that $H(\Omega,\tilde{\Omega})\ge0$ is not clear from this formula, but follows from passivity \cite{LongoMorsella2023,PuszWoronowicz}.)
\end{proof}

More generally, \cite[Thm.~4.5]{Ciolli2020} determines the relative entropy of any two standard coherent states by approximating them by states to which Proposition \ref{prop:weylactionomodularoperators} applies and by using the lower semi-continuity of the relative entropy \cite{Araki1975,Kosaki1986}. The result is then well-defined, but may be infinite. A further generalisation to relative entropies for non-standard states can be found in \cite{Bostelmann2022}. Note that the approximation argument may not apply to the relative modular Hamiltonian as an operator. In such cases it may be useful to consider estimates.

\subsection{Estimates for coherent states}\label{sec:generalcoherentestimates}

In order to apply our estimates to coherent states, we fix two standard subspaces $H_i\subset\Hcal$ with $H_1\subset H_2$, the reference state $\hat{\Omega}=\Omega_{\mathrm{F}}$, which is standard for both $\Rcal(H_i)$, and two coherent states, $\Omega=W(\chi)\Omega_{\mathrm{F}}$ and $\tilde{\Omega}=W(\tilde{\chi})\Omega_{\mathrm{F}}$ with $\chi,\tilde{\chi}\in\Hcal$. Furthermore, we assume that
$\chi=\chi_1+\chi_2$ and $\tilde{\chi}=\tilde{\chi}_1+\tilde{\chi}_2$ with $\chi_2,\tilde{\chi}_2\in H_2$ and
$\chi_1,\tilde{\chi}_1\in H_1'=(\i H_1)^{\perp_r}$, the symplectic adjoint of $H_1$. Then we can decompose the Weyl operator
\begin{align}
\label{eqn:Weyldecompose}
    W(\chi)
  &= \e^{\frac{\i}{2} \Im\langle\chi_1, \chi_2\rangle}
    W(\chi_1) W(\chi_2)
\nexteq\notag
  &= \e^{-\frac{\i}{2} \Im\langle\chi_1, \chi_2\rangle} W(\chi_2) W(\chi_1)
\end{align}
and similarly for $W(\tilde{\chi})$. The Weyl operators $W(\chi_i)$ are unitaries in $\Rcal(H_i)$, so we can combine Theorem \ref{Prop:GeneralEstimate} with Equation \eqref{eqn:WdGamma} to find
\begin{align*}
    K^{(1)}_{\tilde{\Omega}, \Omega}
  &\leq W(\chi_1 + \tilde{\chi}_2) K^{(2)}_{\Omega_{\mathrm{F}}} W(-\chi_1 - \tilde{\chi}_2)
\nexteq
  &= K^{(2)}_{\Omega_{\mathrm{F}}}
    + \frac12 \bigl\langle \chi_1 + \tilde{\chi}_2, K_{H_2}(\chi_1 + \tilde{\chi}_2) \bigr\rangle
    - \phi_{\textup{S}}\bigl( \i K_{H_2}(\chi_1 + \tilde{\chi}_2) \bigr)
\nexteq
    K^{(2)}_{\tilde{\Omega}, \Omega}
  &\geq W(\chi_1 + \tilde{\chi}_2) K^{(1)}_{\Omega_{\mathrm{F}}} W(-\chi_1 - \tilde{\chi}_2)
\nexteq
  &= K^{(1)}_{\Omega_{\mathrm{F}}}
    + \frac12 \bigl\langle \chi_1 + \tilde{\chi}_2, K_{H_1}(\chi_1 + \tilde{\chi}_2) \bigr\rangle
    - \phi_{\textup{S}}\bigl( \i K_{H_1}(\chi_1 + \tilde{\chi}_2) \bigr)
  \eqend{,}
\end{align*}
where the second, resp.~fourth, line holds if $\chi_1+\tilde{\chi}_2$ is in the domain of $K_{H_2}$, resp.~$K_{H_1}$, and the expression acts on a finite particle vector in the domain of $K^{(2)}_{\Omega_{\mathrm{F}}}$, resp,~$K^{(1)}_{\Omega_{\mathrm{F}}}$. If $\chi_1-\tilde{\chi}_1\in H_2'$, then we may apply Corollary \ref{Cor:RelativeModularHamiltonian.CoherentStates} to see that the right-hand side of the first estimate is $K^{(2)}_{\tilde{\Omega},\Omega}$. Similarly, if $\chi_2-\tilde{\chi}_2\in H_1$, then the right-hand side of the second estimate is $K^{(1)}_{\tilde{\Omega},\Omega}$.

For the relative entropy we find in a similar way
\begin{align*}
    H^{(1)}(\Omega,\tilde{\Omega})
  &\leq \frac12 \bigl\langle \chi_2 - \tilde{\chi}_2, K_{H_2}(\chi_2 - \tilde{\chi}_2) \bigr\rangle
\nexteq
    H^{(2)}(\Omega,\tilde{\Omega})
  &\ge q \frac12 \bigl\langle \chi_2 - \tilde{\chi}_2, K_{H_1}(\chi_2 - \tilde{\chi}_2) \bigr\rangle
\end{align*}
if $\chi_2-\tilde{\chi}_2$ is in the domain of $K_{H_2}$ (for the upper bound), resp.~$K_{H_1}$ (for the lower bound). If $\chi_1-\tilde{\chi}_1\in H_2'$, resp.~$\chi_2-\tilde{\chi}_2\in H_1$, then we may apply Corollary \ref{Cor:RelativeEntropy.CoherentStates} to see that the right-hand side of the first, resp.~second, estimate is $H^{(2)}(\Omega,\tilde{\Omega})$, resp.~$H^{(1)}(\Omega,\tilde{\Omega})$.

\section{Coherent excitations in free scalar QFT}\label{sec:coherentscalarQFT}

To investigate the quality of our estimates we will now consider the simplest toy model, the free scalar field in the Minkowski vacuum representation.

In $(d + 1)$-dimensional Minkowski spacetime with signature $-+\cdots+$ we choose the inertial time $0$ surface $\Rbb^d$ as a Cauchy surface. To the Klein-Gordon equation $(-\Box+m^2)\phi=0$ of mass $m\ge0$ we associate the self-adjoint operator $A \coloneq -\Laplace+m^2$ on $L^2(\Rbb^d)$, where $\Laplace$ is the Laplacian. If $m>0$ or $d>1$ we define the real Sobolev spaces
$H^{\left(\pm1/2\right)}(\Rbb^d,\Rbb)$ as the completions of $C_0^{\infty}(\Rbb^d,\Rbb)$ with the inner product given by
$\frac12\langle f,A^{\pm1/2}g\rangle$. We take the one-particle Hilbert space to be
\begin{align}
\Hcal&\coloneq H^{\left(\frac12\right)}(\Rbb^d,\Rbb)\oplus H^{\left(-\frac12\right)}(\Rbb^d,\Rbb)\label{eqn:Hvac}
\end{align}
as a real Hilbert space, endowed with the complex structure $i=\left(\begin{smallmatrix} 0 & -A^{-1/2} \\ A^{1/2} & 0\end{smallmatrix}\right)$.
We may interpret elements of $\Hcal$ as $(f_0,f_1)=\pi(f_0)\hat{\Omega}-\varphi(f_1)\hat{\Omega}$, where $(\varphi,\pi)$ are the distributional initial data of the quantum field $\phi$ and $\hat{\Omega}$ is the Minkowski vacuum state.

It is sometimes convenient to use a covariant formulation in spacetime. To this end we can use the advanced ($-$) and retarded ($+$) fundamental solutions $E^{\pm}$ of the Klein-Gordon operator and write $E=E^--E^+$. We may then define a continuous linear map $D:C_0^{\infty}(\Rbb^{d+1},\Rbb)\to C_0^{\infty}(\Rbb^d,\Rbb)\oplus C_0^{\infty}(\Rbb^d,\Rbb)$ by setting
$Df \coloneq (Ef|_{x_0=0},\partial_0Ef|_{x_0=0})$. We may interpret $Df\in\Hcal$ as $Df=\phi(f)\hat{\Omega}$.

To an open region $O\subset\Rbb^d$ we associate the closed real-linear subspace $H_O\subset\Hcal$ defined by
\begin{align}
H_O&\coloneq \overline{C_0^{\infty}(O,\Rbb)\oplus C_0^{\infty}(O,\Rbb)}\notag
\end{align}
and the von Neumann algebra $\Rcal(H_O)$ on the Fock space $\Fcal_+(\Hcal)$.
Using the Reeh-Schlieder theorem one may show that $H_O\subset\Hcal$ is a standard subspace if $O$ and $\Rbb\setminus\overline{O}$ are both non-empty open sets \cite[Prop.~2.7]{FiglioliniGuido:1989}. Equivalently, $H_O$ can be associated to the spacetime region $V \coloneq D(O)$ that is the domain of dependence of $O$.

\subsection{Reference modular Hamiltonians}
\label{ssec:relmodhamiltonianscoherentstatesfreescalarfields}

In Section \ref{ssec:relmodhamiltonianscoherentstates} we gave expressions for the relative modular Hamiltonian between two coherent states whose one-particle vectors had a difference in $H+H'$ for a standard subspace $H$ and its symplectic complement $H'$. These relative modular Hamiltonians were expressed in terms of the modular Hamiltonian of a reference state. In this section we will review two key examples for free scalar fields, where the modular Hamiltonian is explicitly known.

\begin{example}[Rindler wedge]
The one-particle modular flow of the Rindler wedge in the $x_1$-direction is given by the boost transformation $\Lambda_1^{-2 \pi t}$ \cite{BisognanoWichmann:1975}. Here the base region $R$ is the half space with coordinates $(x^1, \mathvec{x}_{\perp})$, where $x^1 > 0$ and $\mathvec{x}_{\perp}$ denotes the perpendicular spatial coordinates.
The modular flow acting on any element of $\Hcal$ of the form $Df$ with $f\in C_0^{\infty}(\Rbb^{d+1},\Rbb)$ is
\begin{align*}
    \Delta_{H_R}^{\i t} D f(x)
  &= Df\bigl( \Lambda_1^{2 \pi t} x \bigr)
  \eqend{,}
\nexteq
    \Lambda_1^{s} (x^0, x^1, \mathvec{x}_{\perp})
  &= \left(
      x^0 \cosh(s) + x^1 \sinh(s),
      x^0 \sinh(s) + x^1 \cosh(s),
      \mathvec{x}_{\perp}
    \right)
  \eqend{.}
\end{align*}
The relative modular Hamiltonian at one-particle level is given by
\begin{align}
K_{H_R}f&=2\pi\begin{pmatrix}A^{-\frac12}(x^1A-\partial_1)&0\\ 0&A^{\frac12}x^1\end{pmatrix}f\label{eqn:RMHWedge}
\end{align}
for any $f=(f_0,f_1)\in C_0^{\infty}(\Rbb^d,\Rbb)\oplus C_0^{\infty}(\Rbb^d,\Rbb)$. If $f$ is supported away from $\{x^1=0\}=\partial R$, then we find from Corollary \ref{Cor:RelativeEntropy.CoherentStates} that
\begin{align}
    H(W(f)\hat{\Omega}, \hat{\Omega})
  &= \frac{\pi}{2} \int_{x^1 > 0} \left[
      f_0(x^1 A - \partial_1) f_0 + x^1 f_1^2
    \right] \id[d]x
\nnexteq
\label{eqn:REWedge}
  &= \frac{\pi}{2} \int_{x^1 > 0}
      x^1 \left( |\nabla f_0|^2 + m^2 f_0^2 + f_1^2 \right)
    \id[d]x
\nexteq
  &= \pi \int_{x^1 > 0}
       x^1 \left\langle \phi_{\mathrm{S}}(f) \hat{\Omega}, T_{00}(x)|_{x^0 = 0} \phi_{\mathrm{S}}(f) \hat{\Omega} \right\rangle
    \id[d]x
  \eqend{,}
\notag
\end{align}
where $T_{00}(x)$ is the normal-ordered energy density, cf.~\cite{Ciolli2020} for a derivation of this result for more general $f$.
\end{example}

\begin{example}[Massless scalar field on the double cone]
Similarly, we can consider the massless scalar field for a double cone in $(d+1)$-dimensional Minkowski space with base
$B = \{\|\mathvec{x}\| < r\}$ a ball of radius $r > 0$. The modular flow is known to be related to that of the Rindler wedge via a conformal mapping \cite{HislopLongo1982} and the flow acts as \cite[Sec.~V.4.2]{haag2012local},
\begin{align}
\label{eq:DoubleCone.ConformalTransformedTestFunction}
    \Delta_{H_B}^{\i t} D f(x)
  &= D N(2\pi t,x)^{\frac{1-d}{2}}f(T^{2 \pi t} x)
  \eqend{,}
\nexteq
    T^s (x^0, \mathvec{x})
  &= N(s, x^0, \mathvec{x})^{-1} \left(
      x^0 \cosh(s)
      + \frac{r^2 - x^2}{2r} \sinh(s),
      \mathvec{x}
    \right)
  \eqend{,}
\nnexteq
    N(s, x^0, \mathvec{x})
  &\coloneq \frac{x^0}{r} \sinh(s)
    + \frac{r^2 - x^2}{2r^2} \cosh(s)
    + \frac{r^2 + x^2}{2r^2}
  \eqend{,}
\notag
\end{align}
where we wrote $x^2 = -(x^0)^2 + \|\mathvec{x}\|^2$ and the Jacobian factor $N(t,x)^{\frac{1-d}{2}}$ is needed to preserve solutions.
The relative modular Hamiltonian at one-particle level is given by
\begin{align}
K_{H_B}f&=2\pi\begin{pmatrix}A^{-\frac12}\left(\frac{1}{r}x^j\partial_j-\frac{r^2-\|\mathvec{x}\|^2}{2r}\Delta+\frac{d-1}{2r}\right)&0\\ 0&A^{\frac12}\frac{r^2-\|\mathvec{x}\|^2}{2r}
\end{pmatrix}f\label{eqn:RMHDoubleCone}
\end{align}
for any $f=(f_0,f_1)\in C_0^{\infty}(\Rbb^d,\Rbb)\oplus C_0^{\infty}(\Rbb^d,\Rbb)$. If $f$ is supported away from $\{\|\mathvec{x}\|=r\}=\partial B$, then we find from Corollary \ref{Cor:RelativeEntropy.CoherentStates} that
\begin{align}
    H(W(f)\hat{\Omega}, \hat{\Omega})
  &= \frac{\pi}{2} \int_B \left[
      f_0
      \left( -\nabla^j \frac{r^2 - \|\mathvec{x}\|^2}{2 r} \nabla_j + \frac{d - 1}{2 r} \right)
      f_0
      + \frac{r^2 - \|\mathvec{x}\|^2}{2 r}
      f_1^2
    \right] \id[d]x
\nnexteq
\label{eqn:REDoubleCone}
  &= \frac{\pi}{2} \int_B \left[
      \frac{r^2 - \|\mathvec{x}\|^2}{2 r}
      \bigl( |\nabla f_0|^2 + f_1^2 \bigr)
      + \frac{d - 1}{2 r} f_0^2
    \right] \id[d]x
    \eqend{.}
\end{align}
We refer to \cite{LongoMorsella2023} for more details on the massless modular Hamiltonian.
\end{example}

\subsection{Estimates for free scalar fields}\label{ssec:relentropyboundscoherentstates}

We now want to apply Theorem \ref{Prop:GeneralEstimate} to obtain upper and lower bounds on relative modular Hamiltonians. For this purpose we introduce three open regions $O_1\subset O_2\subset O_3\subset\Rbb^d$ in space such that
$\overline{O_i}\subset O_{i+1}$ for $i=1,2$ and $\Rbb^d\setminus\overline{O_3}\not=\emptyset$. The Minkowski vacuum
$\hat{\Omega}$ is then standard for the three local algebras $\Rcal_i \coloneq \Rcal(H_{O_i})$ and therefore also for any coherent states
$\Omega=W(f)\hat{\Omega}$ and $\tilde{\Omega}=W(\tilde{f})\hat{\Omega}$ with
$f,\tilde{f}\in \Hcal=H^{\left(\frac12\right)}(\Rbb^d,\Rbb)\oplus H^{\left(-\frac12\right)}(\Rbb^d,\Rbb)$. For simplicity, we will make the stronger assumption that $f,\tilde{f}\in C_0^{\infty}(\Rbb^d,\Rbb)\oplus C_0^{\infty}(\Rbb^d,\Rbb)$.

For $i=1,2$ we now choose cutoff functions $\eta_i\in C^{\infty}(\Rbb^d,\Rbb)$ with $\eta_i\equiv 1$ on $\overline{O_i}$ and
$\eta_i\equiv 0$ on $\Rbb^d \setminus \overline{O_{i+1}}$.
We will write $\eta_i f = (\eta_i f_0, \eta_i f_1)$, so
$\eta_i f \in C_0^{\infty}(\Reals^d, \Reals) \oplus C_0^{\infty}(\Reals^d, \Reals)$ and the Weyl relations \eqref{eqn:Weylrelations} yield
\begin{align*}
    W(f)
  &= \e^{\frac{\i}{2} \Im\langle f, \eta_i f \rangle} W\bigl( (1 - \eta_i) f \bigr) W(\eta_i f)
\nexteq
  &= \e^{-\frac{\i}{2} \Im\langle f, \eta_i f \rangle} W(\eta_i f) W\bigl( (1 - \eta_i) f \bigr)
\end{align*}
in analogy to \eqref{eqn:Weyldecompose}. As in Section \ref{sec:generalcoherentestimates} we then find
\begin{align*}
    W(g_1) K^{(1)}_{\hat{\Omega}} W(-g_1)
  &\leq K^{(2)}_{\tilde{\Omega}, \Omega}
  \leq W(g_2) K^{(3)}_{\hat{\Omega}} W(-g_2)
  \eqend{,}
\end{align*}
where we introduced $g_i \coloneq (1-\eta_i)f+\eta_i\tilde{f}$. Here the modular Hamiltonian of the vacuum is second quantised, wo we may write $K^{(j)}_{\hat{\Omega}} = \d\Gamma(k^{(j)})$ with the abbreviated notation
\begin{align}
k^{(j)}&\coloneq  K_{H_{O_j}}\eqend{.}\notag
\end{align}
If $g_1$ is in the domain of $k^{(1)}$ and $g_2$ in that of $k^{(3)}$ we may use \eqref{eqn:WdGamma} to find
\begin{align}
\label{eq:RelativeHamiltonians.Bounds}
    K^{(1)}_{\hat{\Omega}}
    + \frac12 \langle g_1, k^{(1)} g_1 \rangle
    - \phi_{\textup{S}}(\i k^{(1)} g_1)
  &\leq K^{(2)}_{\tilde{\Omega}, \Omega}
  \leq K^{(3)}_{\hat{\Omega}}
    + \frac12 \langle g_2, k^{(3)} g_2 \rangle
    - \phi_{\textup{S}}(\i k^{(3)} g_2)
  \eqend{.}
\end{align}

For the relative entropy $H^{(2)}(\Omega,\tilde{\Omega})=\langle\Omega,K^{(2)}_{\tilde{\Omega},\Omega}\Omega\rangle$ we similarly find
\begin{align*}
    \left\langle
      \hat{\Omega},
      W(h_1) K^{(1)}_{\hat{\Omega}} W(-h_1) \hat{\Omega}
    \right\rangle
  &\leq H^{(2)}(\Omega, \tilde{\Omega})
  \leq \left\langle
      \hat{\Omega},
      W(h_2) K^{(3)}_{\hat{\Omega}} W(-h_2) \hat{\Omega}
    \right\rangle
  \eqend{,}
\end{align*}
where $h_i \coloneq g_i-f=\eta_i(\tilde{f}-f)$. If we assume that $h_1$ is in the domain of $k^{(1)}$ and $h_2$ in that of $k^{(3)}$ this yields
\begin{align}
\label{eq:RelativeEntropy.Bounds}
    \frac{1}{2} \langle h_1, k^{(1)} h_1 \rangle
  &\leq H^{(2)}(\Omega, \tilde{\Omega})
  \leq \frac{1}{2} \langle h_2, k^{(3)} h_2 \rangle
  \eqend{.}
\end{align}
If the one-particle modular Hamiltonian can be expressed in terms of an energy density, this gives a kind of Bekenstein bound
\cite{Bekenstein1981,longo2025bekenstein,hollands2025bekenstein}.

For a more compact notation below, we denote the upper, resp.\@ lower, bound in \eqref{eq:RelativeEntropy.Bounds} by $H^+$, resp.\@ $H^-$, so that the difference $H^+ - H^-$ determines the estimate quality. Correspondingly, we write $\eta_+ = \eta_2$, $\eta_- = \eta_1$.
The coherent states are determined by $h_{\pm} \coloneq \eta_{\pm} g$ with $g = \tilde{f} - f = (g_0, g_1) \in \Cinf[0](\Reals^d, \Reals) \oplus \Cinf[0](\Reals^d, \Reals)$.

\subsection{Quality of our general estimates for relative entropies}\label{ssec:qualityofestimates}

To analyse the quality of the relative entropy estimates \eqref{eq:RelativeEntropy.Bounds}, we take the difference between the upper and the lower bound and investigate a suitable limit when the three regions $O_1 \subset O_2 \subset O_3$ are all wedges, 
resp.\@ double cones, for a massive, resp.\@ massless, free scalar field. In both cases we will make use of the following lemma.

\begin{lemma}
\label{lma:CutoffFunction.Integral}
Let $\Scal=\{\eta\in C^{\infty}(\Rbb,\Rbb)\mid \eta\equiv0\mathrm{\ on\ }x\le-1, \eta\equiv1\mathrm{\ on\ }x\ge1\}$ and
$E:\Scal\to\Rbb$ the linear map defined by
\begin{align}
\label{eq:CutoffFunction.Integral}
    E[\eta]
  &\coloneq \int_{-1}^{1} (x + 1) \eta'(x)^2 \id{x}
  \eqend{,}
\end{align}
where the prime denotes the derivative.
Then $\inf_{\eta\in\Scal}E[\eta]=0$.
\end{lemma}
\begin{proof}
It is clear that $E[\eta] \ge 0$ for all $\eta \in \Scal$. If we view the functional $E[\eta]$ as a Lagrangian, then any extremal value is attained at a solution of the Euler-Lagrange equation $-\partial_x(x+1)\partial_x\eta=0$ on the interval $(-1,1)$, i.e., $\eta'(x) = c (x + 1)^{-1}$ for some $c \in \Reals$. Unfortunately, $E[\eta]$ would be infinite for any such solution with $c\not=0$. Nevertheless, we will show that we can approximate such a solution by functions that make the value of $E$ arbitrarily small. 

\begin{figure}
  \centering
  \includegraphics{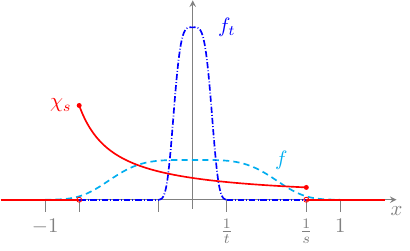}
  \caption{\label{fig:Cutoff.Infimum} The functions $\chi_s$ (red solid curve), $f$ (cyan dashed curve) and $f_t$ (blue dash-dotted curve) used in the proof of Lemma \ref{lma:CutoffFunction.Integral}.}
\end{figure}
We choose a parameter $s > 1$ to restrict the extremal function $\eta'(x)$ to the interval $[-s^{-1}, s^{-1}]$ and normalise it with the constant $c_s = \log(s + 1) - \log(s - 1)$.
I.e., we set $\chi_s(x) \coloneq \bigl( c_s (x + 1) \bigr)^{-1}$ for $|x| < s^{-1}$ and
$\chi_s \equiv 0$ for $|x| \ge s^{-1}$, which yields $\int \chi_s(x) \id{x} = 1$, cf.~Figure~\ref{fig:Cutoff.Infimum}.
Consider any positive function $f \in \Cinf[0](\Reals, \Reals)$ supported in $[-1, 1]$ and normalised, $\int f(x) \d{x} = 1$, which we rescale to $f_t(x) \coloneq t f(t x)$ with a second parameter $t$ such that $t^{-1} \leq 1 - s^{-1}$ or, equivalently, $t \geq \frac{s}{s - 1}$. Then $f_t$ is a smooth function supported in $[- t^{-1}, t^{-1}]$ and the convolution $\chi_s * f_t$ is a smooth function supported in $[-1, 1]$. All of these functions are normalised by construction. We now define
\begin{align*}
    \eta_{s,t}(x)
  &\coloneq \int_{-\infty}^x (\chi_s*f_t)(y) \id{y}
  \eqend{,}
\end{align*}
which is an element in $\Scal$, and we consider
\begin{align*}
    E[\eta_{s,t}]
  &= \int_{-1}^1 \iint_{\Reals^2} (x + 1)
      \chi_s\left( x - \frac{y}{t} \right) f(y)
      \chi_s\left( x - \frac{z}{t} \right) f(z)
    \id{z} \id{y} \id{x}
  \eqend{.}
\end{align*}
Because $\chi_s$ is continuous at $x \neq \pm s^{-1}$, the integrand converges pointwise to $(x + 1) \chi_s(x)^2 f(y) f(z)$ as $t \to \infty$ at any $x\not=\pm s^{-1}$.
Moreover, as $\|\chi_s\|_{\infty}=\frac{s}{c_s(s-1)}$ the integrand is bounded by $\frac{2 s^2}{c_s^2(s-1)^2}f(y)f(z)$, which is integrable on
$[-1,1]\times\Rbb^2$, so by the Dominated Convergence Theorem and the normalisation of all the functions, we find
\begin{align*}
    \lim_{t \to \infty} E[\eta_{s, t}]
  &= \int_{-1}^1 (x + 1) \chi_s(x)^2
      \iint_{\Reals^2} f(y) f(z)
      \id{z} \id{y}
    \id{x}
\nexteq
  &= \int_{-1}^1 (x + 1) \chi_s(x)^2 \id{x}
\nexteq
  &= \frac{1}{c_s}
    \int_{-1}^1 \chi_s(x) \id{x}
\nexteq
  &= \left( \log\frac{s + 1}{s - 1} \right)^{-1}
  \eqend{.}
\end{align*}
The result vanishes as $s \to 1^+$, so we can make $E[\eta_{s, t}]$ as small as we like.
\end{proof}
Note that for any $\eta_+ \in \mathcal{S}$ the function $\eta_-(x) \coloneq 1 - \eta_+(-x)$ is also in $\mathcal{S}$, so that the lemma also holds when we replace $\eta'(x)$ by $\eta'(-x)$ in \eqref{eq:CutoffFunction.Integral}.
The derivative $\eta'$ of any function $\eta \in \mathcal{S}$ is a mollifier, i.e.\@ it has compact support in $[-1, 1]$, is normalised $\int_{-1}^{1} \eta'(x) \id{x} = 1$,
and the functions $\eta_{\varepsilon}'(x) \coloneq \frac{1}{\varepsilon} \eta'\left( \frac{x}{\varepsilon} \right)$ are a representation of the Dirac distribution in the limit $\varepsilon \to 0^+$.

\subsubsection{Nested Rindler wedges}
\label{sec:RelativeEntropy.EnergyDensity.RightWedges}

As a first example, we consider coherent states over nested Rindler wedge regions $W=D(R)$, $W_{\pm}=D(R_{\pm})$ with bases $R=\{x^1>0\}$, $R_{\pm} \coloneq \{x^1 - x_{\pm} > 0\}$ illustrated in Figure~\ref{fig:NestedRightWedges}.
\begin{figure}
  \centering
  \includegraphics{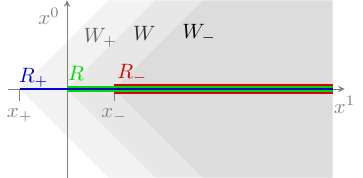}
  \caption{\label{fig:NestedRightWedges} Illustration of the nested wedge regions $W_-\subset W\subset W_+$ and their respective half spaces $R_- \subset R \subset R_+$ in the Cauchy hypersurface $\Reals^d$. Here, only the $x^0,x^1$ directions are shown and perpendicular spatial coordinates are suppressed.}
\end{figure}
It will be convenient to introduce the corresponding weight functions
\begin{align}
\label{eq:RightWedges.WeightFunctions}
    \beta(\mathvec{x})
  &= x^1
  \eqend{,}
&
    \beta_{\pm}(\mathvec{x})
  &\coloneq x^1 - x_{\pm}
  \eqend{.}
\end{align}

The upper and lower bounds $H^{\pm}$ in \eqref{eq:RelativeEntropy.Bounds} can then be expressed using \eqref{eqn:RMHWedge} as
\begin{align}
\label{eq:RightWedges.RelativeEntropy.Bounds}
    H^{\pm}
  &= \frac{\pi}{2} \int_{R_+}
      \beta_{\pm} \Bigl[
        (\nabla \eta_{\pm} g_0)^2
        + m^2 \eta_{\pm}^2 g_0^2
        + \eta_{\pm}^2 g_1^2
      \Bigr]
    \id[d]{x}
\nnexteq
  &\pickindent{=} \frac{\pi}{2} \int_{R_+}
      \beta_{\pm} \eta_{\pm}^2
      \Bigl[
        (\nabla g_0)^2
        + m^2 g_0^2
        + g_1^2
      \Bigr]
    \id[d]{x}
    + \frac{\pi}{2} \int_{R_+ \setminus R_-}
      \beta_{\pm} \eta_{\pm} (\nabla \eta_{\pm})
      \cdot (\nabla g_0^2)
    \id[d]{x}
  \eqbreakr
    + \frac{\pi}{2} \int_{R_+ \setminus R_-}
      \beta_{\pm} (\nabla \eta_{\pm})^2
      g_0^2
    \id[d]{x}
  \eqend{,}
\end{align}
where we have restricted the integrals to $R_+$ or $R_+ \setminus R_-$, because $\supp \eta_{\pm} \subset R_+$ and $\supp \nabla \eta_{\pm} \subset R_+ \setminus R_-$.
We also simplified the expression with the identity $2 g_0 \nabla g_0 = \nabla g_0^2$.

\begin{theorem}
\label{thm:RelativeEntropy.Bounds.RightWedges}
For any $f, \tilde{f} \in C_0^{\infty}(\Reals^d, \Reals) \oplus C_0^{\infty}(\Reals^d, \Reals)$, the relative entropy between the coherent excitations $\Omega = W(f) \hat{\Omega}$ and $\tilde{\Omega} = W(\tilde{f}) \hat{\Omega}$ on the Rindler wedge $W = D(R)$ is 
\begin{align}
\label{eq:RightWedges.RelativeEntropy}
    H(\Omega, \tilde{\Omega})
  &= \frac{\pi}{2} \int_{R}
      \beta \Bigl[
        (\nabla g_0)^2
        + m^2 g_0^2
        + g_1^2
      \Bigr]
    \id[d]{x}
\end{align}
with $(g_0,g_1):=\tilde{f}-f$, which is approximated arbitrarily well by $H^+$ from above and $H^-$ from below.
\end{theorem}
\proof
For the simplification of the bounds \eqref{eq:RightWedges.RelativeEntropy.Bounds}, we place the two bounding half spaces at equal distances $x_{\pm} = \mp 2 \varepsilon$, see also Figure~\ref{fig:NestedRightWedges}.
We construct corresponding cutoff functions $\chi_{\pm}$ that vary only in the $x^1$-direction from arbitrary $\tilde{\eta}_{\pm} \in \mathcal{S}$ as considered in \autoref{lma:CutoffFunction.Integral},
\begin{align}
\label{eq:RightWedges.RelativeEntropy.BoundsSymmetric.Cutoffs}
    \eta_{\pm}(\mathvec{x})
  &\coloneq \tilde{\eta}_{\pm}\left( \frac{x^1}{\varepsilon} \pm 1 \right)
  \eqend{.}
\end{align}

Since the weights \eqref{eq:RightWedges.WeightFunctions} and the cutoffs \eqref{eq:RightWedges.RelativeEntropy.BoundsSymmetric.Cutoffs} only vary with $x^1$, we integrate $g_0$ over all other spatial coordinates to get the function
\begin{align*}
    \tau_0(x^1)
  &\coloneq \int_{(x^2, x^3, \dots) \in \Reals^{d - 1}}
      g_0^2
    \id[d - 1]{x}
  \eqend{,}
\end{align*}
which is in $C_0^{\infty}(\Reals, \Reals)$.
We substitute $x = x^1 / \varepsilon$ in the second and third integrals of the bounds \eqref{eq:RightWedges.RelativeEntropy.Bounds} and get
\begin{align*}
    H^{\pm}
  &\pickindent{=} \frac{\pi}{2} \int_{R_+}
      \beta_{\pm} \eta_{\pm}^2
      \Bigl[
        (\nabla g_0)^2
        + m^2 g_0^2
        + g_1^2
      \Bigr]
    \id[d]{x}
    + \frac{\pi \varepsilon}{2} \int_{-2}^{2}
      (x \pm 2) \tilde{\eta}_{\pm}(x \pm 1) \tilde{\eta}_{\pm}'(x + 1)
      \tau_0'(\varepsilon x)
    \id{x}
  \eqbreakr
    + \frac{\pi}{2} \int_{-2}^{2}
      (x \pm 2) \tilde{\eta}_{\pm}'(x \pm 1)^2
      \tau_0(\varepsilon x)
    \id{x}
  \eqend{.}
\end{align*}
By the Dominated Convergence Theorem the first integral converges to \eqref{eq:RightWedges.RelativeEntropy} in the limit 
$\varepsilon \to 0^+$, because a dominating function is readily found and $\beta_{\pm} \to \beta$ and $\eta_{\pm}(\mathvec{x}) \to \theta(x^1)$ converge pointwise, where $\theta$ is the Heaviside step function.
Similarly, the second term converges to $0$ and the third term to 
\begin{align*}
    \frac{\pi}{2} \tau_0(0)
    \int_{-2}^{2}
      (x \pm 2) \tilde{\eta}_{\pm}'(x \pm 1)^2
    \id{x}
  &= \pm \frac{\pi}{2} \tau_0(0) \int_{-1}^{1}
      (y + 1) \tilde{\eta}_{\pm}'(\pm y)^2
    \id{y}
  \eqend{,}
\end{align*}
where we used the substitution $\pm y \coloneq x \pm 1$ and the support properties of $\tilde{\eta}_{\pm}'$.
This integral can be made arbitrarily small by Lemma~\ref{lma:CutoffFunction.Integral} and the remarks below it.
\qed

We note that the regularity assumptions on $f,\tilde{f}$ can be weakened considerably without affecting the proof.

\subsubsection{Nested double cones for massless fields}
\label{sec:RelativeEntropy.EnergyDensity.DoubleCones}

In the case of double cones, the arguments are very similar, but now the bases $B_- \subset B \subset B_+$ are $d$-dimensional balls with radii $0 < r_- < r < r_+$, for simplicity centred around the coordinate origin of $\Reals^d$.
The corresponding weight functions are
\begin{align}
\label{eq:DoubleCones.WeightFunctions}
    \beta(\mathvec{x})
  &\coloneq \frac{r^2 - \|\mathvec{x}\|^2}{2 r}
  \eqend{,}
&
    \beta_{\pm}(\mathvec{x})
  &\coloneq \frac{r_{\pm}^2 - \|\mathvec{x}\|^2}{2 r_{\pm}}
  \eqend{.}
\end{align}

The upper and lower bounds $H^{\pm}$ in \eqref{eq:RelativeEntropy.Bounds} can then be expressed using 
\eqref{eqn:RMHDoubleCone} as
\begin{align}
\label{eq:DoubleCones.RelativeEntropy.Bounds}
    H^{\pm}
  &= \frac{\pi}{2} \int_{B_+}
      \beta_{\pm}
      \bigl[
        (\nabla \eta_{\pm} g_0)^2
        + \eta_{\pm}^2 g_1^2
      \bigr]
    \id[d]{x}
    + \frac{\pi (d - 1)}{4 r_{\pm}} \int_{B_+}
      \eta_{\pm}^2 g_0^2
    \id[d]{x}
\nnexteq
  &\pickindent{=} \frac{\pi}{2} \int_{B_+}
      \beta_{\pm} \eta_{\pm}^2
      \bigl[
        (\nabla g_0)^2
        + g_1^2
      \bigr]
    \id[d]{x}
    + \frac{\pi(d - 1)}{4 r_{\pm}} \int_{B_+}
      \eta_{\pm}^2 g_0^2
    \id[d]{x}
  \eqbreakr
    + \frac{\pi}{2} \int_{B_+ \setminus B_-}
      \beta_{\pm} \eta_{\pm} (\nabla \eta_{\pm})
      \cdot (\nabla g_0^2)
    \id[d]{x}
    + \frac{\pi}{2} \int_{B_+ \setminus B_-}
      \beta_{\pm} (\nabla \eta_{\pm})^2
      g_0^2
    \id[d]{x}
  \eqend{,}
\end{align}
where we have restricted the integrals to $B_+$ and $B_+ \setminus B_-$, because $\supp \eta_{\pm} \subset B_+$ and $\supp \nabla \eta_{\pm} \subset B_+ \setminus B_-$, respectively.

\begin{theorem}
\label{thm:RelativeEntropy.Bounds.DoubleCones}
In the massless case, for any $f, \tilde{f} \in C_0^{\infty}(\Reals^d, \Reals) \oplus C_0^{\infty}(\Reals^d, \Reals)$, the relative entropy between the coherent excitations $\Omega = W(f) \hat{\Omega}$ and $\tilde{\Omega} = W(\tilde{f}) \hat{\Omega}$ on the double cone $D(B)$ with $B$ a ball of radius $r$ is
\begin{align}
\label{eq:DoubleCones.RelativeEntropy}
    H(\Omega, \tilde{\Omega})
  &= \frac{\pi}{2} \int_{B}
    \left(
      \beta \bigl[
        (\nabla g_0)^2
        + g_1^2
      \bigr]
      + \frac{d - 1}{2 r}
      g_0^2
    \right) \id[d]{x}
\end{align}
with $(g_0,g_1):=\tilde{f}-f$, which is approximated arbitrarily well by $H^+$ from above and $H^-$ from below.
\end{theorem}
\proof
Similar to the Rindler wedges in Section~\ref{sec:RelativeEntropy.EnergyDensity.RightWedges}, we approximate the middle region $B$ symmetrically by balls $B_{\pm}$ with radii $r_{\pm} = r \pm 2 \varepsilon$ for $0 < \varepsilon < \frac{r}{2}$ and use cutoff functions that are spherically symmetric,
\begin{align}
\label{eq:DoubleCones.RelativeEntropy.BoundsSymmetric.Cutoffs}
    \eta_{\pm}(\mathvec{x})
  &\coloneq \tilde{\eta}_{\pm}\left( \frac{r - \|\mathvec{x}\|}{\varepsilon} \pm 1 \right)
  \eqend{.}
\end{align}
We integrate first over the unit sphere $S^{d-1}$ 
\begin{align*}
    \tau_0\bigl( \|\mathvec{x}\| \bigr)
  &\coloneq \int_{S^{d - 1}}
      g_0^2
    \id\vol_{S^{d - 1}}
  \eqend{,}
\end{align*}
which is a function in $C_0^{\infty}([0,\infty), \Reals)$.
We substitute $s = (r-\|\mathvec{x}\|) / \varepsilon$ in the second and third integrals of $H^{\pm}$ and find
\begin{align*}
   H^{\pm}
  &\pickindent{=} \frac{\pi}{2} \int_{B}
    \eta_{\pm}^2
      \left(
        \beta_{\pm}
        \bigl[
          (\nabla g_0)^2
          + g_1^2
        \bigr]
        + \frac{d - 1}{2 r_{\pm}}
        g_0^2
      \right)
      \id[d]{x}
  \eqbreakr
    + \frac{\pi\varepsilon}{2} \int_{-2}^{2}
      \frac{r_{\pm}^2 - (r - \varepsilon s)^2}{2 r_{\pm}\varepsilon}
      \tilde{\eta}_{\pm}(s \pm 1)
      \tilde{\eta}_{\pm}'(s \pm 1)
      \tau_0'(r - \varepsilon s)
      (r - \varepsilon s)^{d - 1}
    \id{s}
  \eqbreakr
    + \frac{\pi}{2} \int_{-2}^{2}
      \frac{r_{\pm}^2 - (r - \varepsilon s)^2}{2 r_{\pm}\varepsilon}
      \tilde{\eta}_{\pm}'(s \pm 1)^2
      \tau_0(r - \varepsilon s)
      (r - \varepsilon s)^{d - 1}
    \id{s}
  \eqend{.}
\end{align*}
Again, by the Dominated Convergence Theorem the first integral converges to \eqref{eq:DoubleCones.RelativeEntropy} in the limit 
$\varepsilon \to 0^+$, because a dominating function is readily found and $\beta_{\pm} \to \beta$ and $\eta_{\pm}(\mathvec{x}) \to \theta\bigl( r - \|\mathvec{x}\| \bigr)$ converge pointwise. 
Similarly, the second term converges to $0$ and the third term to
\begin{align*}
    \frac{\pi}{2} \tau_0(r) r^{d - 1}
    \int_{-2}^{2}
      (s \pm 2)
      \tilde{\eta}_{\pm}'(s \pm 1)^2
    \id{s}
  &= \pm \frac{\pi}{2} \tau_0(r) r^{d - 1}
    \int_{-1}^{1}
      (y + 1) \tilde{\eta}_{\pm}'(\pm y)^2
    \id{y}
  \eqend{,}
\end{align*}
where we used the fact that
$\frac{r_{\pm}^2 - (r - \varepsilon s)^2}{2 r_{\pm}\varepsilon}\to (s \pm 2)$ pointwise as $\varepsilon\to0$, the substitution $\pm y = s \pm 1$ and the support properties of $\tilde{\eta}_{\pm}$. This integral is the same as in the proof of Theorem \ref{thm:RelativeEntropy.Bounds.RightWedges}, so it can be made arbitrarily small.
\qed

\section{Conclusions and discussion}\label{sec:conclusions}

Within the scope of this paper, we have presented a scheme for general AQFTs that allows us to estimate relative modular Hamiltonians, and hence also relative entropies, between suitable pairs of states from above, resp.~below, in terms of modular Hamiltonians of a reference state on a larger, resp.~smaller, region. This is the content of our main result, Theorem \ref{Prop:GeneralEstimate}.
We have argued that in a typical QFT setting, where local algebras are of type III, the class of pairs of states that are susceptible to our scheme is quite large, due to a result of Connes and St{\o}rmer \cite{Connes1978}. It remains an open question whether all pairs of quasi-equivalent, standard states can be handled in this way. We expect that this question can only be addressed by adapting the methods of \cite{Connes1978}.

An interesting aspect of our estimates arises when one of the states of interest is the reference state. In that case, the estimates apply only if the transition between the states can be accomplished by a unitary operation that does not allow 
superluminal signalling in the context of Sorkin's impossible measurements. It is unclear whether the converse of this implication is also valid, i.e., whether every non-signalling unitary $u$ and every vector $\Omega$ that is standard for the algebras involved, the pair $\Omega$ and $u\Omega$ is susceptible to our estimates. For type I factors, we have shown in Counterexample~\ref{ex:nonproductsignallingunitary} that this converse implication fails. However, the arguments used there do not extend to the type III case.

Let us briefly point out that our general estimates can be combined very nicely with spacetime deformation arguments as in \cite{KS2009}.
I.e., we do not need to insist that $V_1\subset V_2\subset V_3$, as long as $V_1\subset D(V_2)$ and $D(V_3)\supset V_2$ and the theory satisfies the time-slice axiom. This could lead to Bekenstein-type bounds on relative entropies, if one starts e.g.~with a Minkowski vacuum, whose relative modular Hamiltonian on a double cone is related to the energy density, and one deforms the theory and the states into the future.

In the context of coherent states on CCR algebras, our general estimates simplify and for the particular case of free scalar fields all pairs of coherent excitations of the vacuum are susceptible to our estimates. In Section \ref{sec:coherentscalarQFT} we combined our estimates with a squeezing argument to recover a known formula for the relative entropy on a wedge or a double cone (in the massless case), cf.~\cite{CasiniGrilloPontello:2019,Ciolli2020,LongoMorsella2023,Bostelmann2022}. This implies in particular that the relative entropy can be approximated arbitrarily well by our estimates, suggesting that they might not lose information also in the general case. Analogous limiting arguments for the relative modular Hamiltonians are more difficult to formulate, due to the unboundedness of the operators involved, which complicate the notion of convergence. One possibility would be to investigate convergence of the estimates in \eqref{eq:RelativeHamiltonians.Bounds} in the graph norm, but we have not pursued this here.

A natural question to ask is whether it is possible to extend our estimates beyond the case of coherent states on CCR algebras. For pairs of general quasi-free states \(\omega,\tilde{\omega}\) the situation is already more complicated. If the states are quasi-equivalent and standard for the region \(V_2\), as in Section~\ref{ssec:applicationinQFT}, then we may implement both states by standard vectors \(\Omega\), \(\tilde{\Omega}\) in the same Hilbert space and there is a Bogolyubov transformation \(b\), implemented by a unitary \(u_b\) satisfying \(\tilde{\Omega}=u_b\Omega\). Bogolyubov unitaries act on Weyl operators (generating local algebras) as \(u_b W(f) u_b^*=W(bf)\), say for \(f\) supported in some bounded region, so they will not preserve generic local algebras and may in principle allow superluminal signalling. It seems difficult to determine in general whether the pair of states is susceptible to our estimates.
(See also
\cite{lashkari2021modular} for results on states with a single squeezed mode.)

\bigskip

\textbf{Acknowledgements}

A.C.\@ and K.S.\@ would like to thank the organisers and the participants of the Workshop \enquote{When Quantum Field Theory meets Quantum Information} at Universidad Carlos III de Madrid (Leganés, Spain) in December 2025, where initial results were presented and their relation to Sorkin's paradox first came up.

K.S.\@ gratefully acknowledges the support by the Heisenberg Programme of the Deutsche Forschungsgemeinschaft (DFG) through the project \enquote{Mathematical Features of the Stress Tensor in Quantum Field Theory in Curved Spacetime}
(SA 3220/1-1).
C.M.\@ would like to acknowledge the Riemann Center for Geometry and Physics at the Leibniz University Hannover for the award of a Riemann Fellowship to conduct this research.
A.C.\@ is affiliated with GNFM–INDAM (Istituto Nazionale di Alta Matematica).

\medskip

\textbf{Declarations}

\emph{Data availability statement:}
The authors declare that no data sets were used or created for this research.

\emph{Competing interests:} The authors have no competing interests to declare.

\printbibliography

@article{Araki1975,
title={Relative Entropy of States of von Neumann Algebras},
author={Araki, H.},
journal={Publ. Res. Inst. Math. Sci.},
volume={11},
number={3},
pages={809-833},
year={1975}
}

@article{Araki1977,
title={Relative Entropy for States of von Neumann Algebras II},
author={Araki, H.},
journal={Publ. Res. Inst. Math. Sci.},
volume={13},
number={1},
pages={173-192},
year={1977}
}

@article{Connes1978,
title = {Homogeneity of the state space of factors of type III$_1$},
journal = {Journal of Functional Analysis},
volume = {28},
number = {2},
pages = {187-196},
year = {1978},
%issn = {0022-1236},
%doi = {https://doi.org/10.1016/0022-1236(78)90085-X},
%url = {https://www.sciencedirect.com/science/article/pii/002212367890085X},
author = {Connes, A. and Størmer, E.},
}

@book{haag2012local,
	title={Local quantum physics: Fields, particles, algebras},
	author={Haag, R.},
	year={2012},
	publisher={Springer Science \& Business Media}
}

@book{KadisonRingroseI,
  title={Fundamentals of the Theory of Operator Algebras. Volume I},
  author={Kadison, R. V. and Ringrose, J. R.},
  year={1997},
  publisher={American Mathematical Society}
}

@book{KadisonRingroseII,
	title={Fundamentals of the Theory of Operator Algebras. Volume II},
	author={Kadison, R. V. and Ringrose, J. R.},
	year={1997},
	publisher={American Mathematical Society}
}

@book{TakesakiI,
	title={Theory of Operator Algebras I},
	author={Takesaki, M.},
	year={1979},
	publisher={Springer}
}

@book{TakesakiII,
	title={Theory of Operator Algebras II},
	author={Takesaki, M.},
	year={2003},
	publisher={Springer}
}

@book{HollandsKS2018,
	title={Entanglement Measures and Their Properties in Quantum Field Theory},
	author={Hollands, S. and Sanders, K.},
	year={2018},
	publisher={Springer},
    %doi={10.1007/978-3-319-94902-4}
}

@book{ReedSimonII,
	title={Methods of Modern Mathematical Physics II},
    author={Reed, M. and Simon, B.},
	year={1975},
	publisher={Academic Press},
}

@article{HislopLongo1982,
	title={Modular structure of the local algebras associated with the free massless scalar field theory},
	author={Hislop, P. D. and Longo, R.},
	journal={Communications in Mathematical Physics},
	volume={84},
    number={1},
    pages={71-85},
	year={1982}
}

@article{LongoMorsella2023,
	title={The Massless Modular Hamiltonian},
	author={Longo, R. and Morsella, G.},
	journal={Communications in Mathematical Physics},
	volume={400},
	number={2},
	pages={1181-1201},
	year={2023}
}

@article{BisognanoWichmann:1975,
    author = "Bisognano, J. J. and Wichmann, E. H.",
    title = "{On the duality condition for a Hermitian scalar field}",
    %doi = "10.1063/1.522605",
    journal = {Journal of Mathematical Physics},
    volume = "16",
    number = "4",
    pages = "985--1007",
    year = "1975",
    publisher = "American Institute of Physics"
}

@article{BisognanoWichmann1976,
    author = {Bisognano, J. J. and Wichmann, E. H.},
    title = {On the duality condition for quantum fields},
    journal = {Journal of Mathematical Physics},
    volume = {17},
    number = {3},
    pages = {303-321},
    year = {1976}
}

@article{araki1982quasi,
  title={On quasi-equivalence of quasifree states of the canonical commutation relations},
  author={Araki, H. and Yamagami, S.},
  journal={Publications of the Research Institute for Mathematical Sciences},
  volume={18},
  number={2},
  pages={703-758},
  year={1982},
  publisher={Research Institute for Mathematical Sciences}
}

@article{longo2025bekenstein,
  title={A Bekenstein-type bound in QFT},
  author={Longo, R.},
  journal={Communications in Mathematical Physics},
  volume={406},
  number={5},
  pages={95},
  year={2025},
  publisher={Springer}
}

@article{hollands2025bekenstein,
author = {Hollands, S. and Longo, R.},
year = {2025},
%month = {05},
pages = {},
title = {Bekenstein Bound for Approximately Local Charged States},
journal = {Reviews in Mathematical Physics},
%doi = {10.1142/S0129055X24610087}
}

@article{Bekenstein1981,
  title = {Universal upper bound on the entropy-to-energy ratio for bounded systems},
  author = {Bekenstein, J. D.},
  journal = {Phys. Rev. D},
  volume = {23},
  pages = {287-298},
  numpages = {0},
  year = {1981},
  publisher = {American Physical Society}
}

@article{Witten:RevModPhys.90.045003,
  title = {APS Medal for Exceptional Achievement in Research: Invited article on entanglement properties of quantum field theory},
  author = {Witten, E.},
  journal = {Rev. Mod. Phys.},
  volume = {90},
  %issue = {4},
  pages = {045003},
  numpages = {38},
  year = {2018},
  publisher = {American Physical Society},
  %doi = {10.1103/RevModPhys.90.045003},
  %url = {https://link.aps.org/doi/10.1103/RevModPhys.90.045003}
% month = {Oct}
}

@incollection{witten2022does,
  title={Why does quantum field theory in curved spacetime make sense? And what happens to the algebra of observables in the thermodynamic limit?},
  author={Witten, E.},
  booktitle={Dialogues Between Physics and Mathematics: CN Yang at 100},
  pages={241-284},
  year={2022},
  publisher={Springer}
}

@inproceedings{longo2008real,
  title={Real Hilbert subspaces, modular theory, \(SL (2, \mathbb{R})\) and CFT},
  author={Longo, R.},
  booktitle={Von Neumann algebas in Sibiu: Conference Proceedings},
  pages={33--91},
  year={2008}
}

@article{BaezFritz2014,
    AUTHOR = {Baez, J. C. and Fritz, T.},
     TITLE = {A {B}ayesian characterization of relative entropy},
   JOURNAL = {Theory Appl. Categ.},
  FJOURNAL = {Theory and Applications of Categories},
    VOLUME = {29},
      YEAR = {2014},
     PAGES = {No. 16, 422--457},
      %ISSN = {1201-561X},
   %MRCLASS = {94A17 (18B99 62F15)},
%  MRNUMBER = {3246591},
%MRREVIEWER = {Samuel\ Cheng},
}

@article{SEWELL1982201,
title = {Quantum fields on manifolds: PCT and gravitationally induced thermal states},
journal = {Annals of Physics},
volume = {141},
number = {2},
pages = {201-224},
year = {1982},
%issn = {0003-4916},
%doi = {https://doi.org/10.1016/0003-4916(82)90285-8},
%url = {https://www.sciencedirect.com/science/article/pii/0003491682902858},
author = {G. L. Sewell},
}

@article{PuszWoronowicz,
	author = "W. Pusz and S. L. Woronowicz",
	journal = "Communications in Mathematical Physics",
	number = "",
	pages = "273--290",
	publisher = "Springer",
	title = "Passive states and KMS states for general quantum systems.",
	volume = "58",
	year = "1978"
}

@article{ReehSchlieder:1961,
    author = {Reeh, H. and Schlieder, S.},
    title = {{Bemerkungen zur Unit{\"a}r{\"a}quivalenz von Lorentzinvarianten Feldern}},
    %doi = "10.1007/BF02787889",
    journal = "Il Nuovo Cim.",
    volume = "22",
    number = "5",
    pages = "1051--1068",
    year = "1961"
}

@article{Longo:2019,
    author = "Longo, R.",
    title = "{Entropy of coherent excitations}",
    %eprint = "1901.02366",
    %archivePrefix = "arXiv",
    %primaryClass = "math-ph",
    %doi = "10.1007/s11005-019-01196-6",
    journal = "Lett. Math. Phys.",
    volume = "109",
    number = "12",
    pages = "2587--2600",
    year = "2019"
}

@article{CasiniGrilloPontello:2019,
    author = "Casini, H. and Grillo, S. and Pontello, D.",
    title = "{Relative entropy for coherent states from Araki formula}",
    %eprint = "1903.00109",
    %archivePrefix = "arXiv",
    %primaryClass = "hep-th",
    %doi = "10.1103/PhysRevD.99.125020",
    journal = "Phys. Rev. D",
    volume = "99",
    number = "12",
    pages = "125020",
    year = "2019"
}

@article{Ciolli2020,
    author = "Ciolli, F. and Longo, R. and Ruzzi, G.",
    title = "{The Information in a Wave}",
    %doi = "10.1007/s00220-019-03593-3",
    journal = "Communications in Mathematical Physics",
    volume = "379",
    number = "3",
    pages = "979--1000",
    year = "2020"
}

@article{Bostelmann2022,
 author = {Bostelmann, H. and Cadamuro, D. and Del Vecchio, S.},
 journal = {Communications in Mathematical Physics},
 volume = {389},
 number = {1},
 pages = {661--691},
 %doi = {10.1007/s00220-021-04249-x},
 title = {Relative Entropy of Coherent States on General CCR Algebras},
 year = {2022}
}

@article{FiglioliniGuido:1989,
    author = "Figliolini, F. and Guido, D.",
    title = "{The Tomita Operator for the free scalar field}",
    journal = "Ann. Inst. H. Poincar\'e",
    volume = "51",
    number = "4",
    pages = "419--435",
    year = "1989",
    publisher = "JSTOR"
}

@article{Kosaki1986,
 %ISSN = {03794024, 18417744},
 %URL = {http://www.jstor.org/stable/24714805},
 author = {Kosaki, H.},
 journal = {Journal of Operator Theory},
 number = {2},
 pages = {335--348},
 publisher = {Theta Foundation},
 title = {Relative entropy of states: a variational expression},
 volume = {16},
 year = {1986}
}

@book{brattelirobinson1997,
	author = {Bratteli, O. and Robinson, D. W. },
	title = { Operator Algebras and Quantum Statistical Mechanics 2: Equilibrium States. Models in Quantum Statistical Mechanics},
	edition = { 2nd ed. },
	%isbn = { 3540614435 },
	publisher = { Springer Berlin ; London },
	pages = { xiii, 517 pages ; },
	year = { 1997 },
	type = { Book },
	language = { English }
}

@book{Simon2019MatrixMonotone,
	title={Loewner's Theorem on Monotone Matrix Functions},
	author={Simon, B.},
	year={2019},
	publisher={Springer Cham}
}

@article{jubb2022causal,
  title={Causal state updates in real scalar quantum field theory},
  author={Jubb, I.},
  journal={Physical Review D},
  volume={105},
  number={2},
  pages={025003},
  year={2022},
  publisher={APS}
}

@article{borsten2021impossible,
  title={Impossible measurements revisited},
  author={Borsten, L. and Jubb, I. and Kells, G.},
  journal={Physical Review D},
  volume={104},
  number={2},
  pages={025012},
  year={2021},
  publisher={APS}
}

@inproceedings{Sorkin:1993,
  author = "Sorkin, R. D.",
  title = "{Impossible measurements on quantum fields}",
  booktitle = "{Directions in General Relativity: An International Symposium in Honor of the 60th Birthdays of Dieter Brill and Charles Misner}",
  %eprint = "gr-qc/9302018",
  %archivePrefix = "arXiv",
  %reportNumber = "SU-GP-92-12-4A",
  month = "2",
  year = "1993",
  pages = "293–-305",
  publisher = "Cambridge University Press"
}

@article{lashkari2021modular,
  title={Modular flow of excited states},
  author={Lashkari, N. and Liu, H. and Rajagopal, S.},
  journal={Journal of High Energy Physics},
  volume={2021},
  number={9},
  pages={1--50},
  year={2021},
  publisher={Springer}
}

@article{papageorgiou2024eliminating,
  title={Eliminating the ‘impossible’: Recent progress on local measurement theory for quantum field theory},
  author={Papageorgiou, M. and Fraser, D.},
  journal={Foundations of Physics},
  volume={54},
  number={3},
  pages={26},
  year={2024},
  publisher={Springer}
}

@article{buchholz1978,
author = {Buchholz, D.},
year = {1978},
month = {01},
pages = {146-153},
title = {On the structure of local quantum fields with non-trivial interaction},
volume = {Proceedings, Leipzig 1977},
journal = {Operator Algebras, Ideals and Their Applications in Theoretical Physics}
}

@article{la2025fermionic,
  title={The fermionic massless modular hamiltonian},
  author={La Piana, F. and Morsella, G.},
  journal={Communications in Mathematical Physics},
  volume={406},
  number={4},
  pages={81},
  year={2025},
  publisher={Springer}
}

@article{sanders2015construction,
  title={On the construction of Hartle--Hawking--Israel states across a static bifurcate Killing horizon},
  author={Sanders, K.},
  journal={Letters in Mathematical Physics},
  volume={105},
  number={4},
  pages={575--640},
  year={2015},
  publisher={Springer}
}

@article{KS2009,
  title={On the Reeh-Schlieder property in curved spacetime},
  author={Sanders, K.},
  journal={Communications in Mathematical Physics},
  volume={288},
  number={1},
  pages={271--285},
  year={2009},
  publisher={Springer}
}

\end{document}